\documentclass[11pt]{amsart}
\usepackage{times}
\usepackage[pdftex]{graphicx,color}
\usepackage{amssymb,amsfonts,amsmath,amsthm}
\usepackage{float}
\usepackage{color}

\newtheorem{theorem}{Theorem}
\newtheorem{proposition}{Proposition}
\newtheorem{lemma}{Lemma}

\newtheorem{remark}{Remark}
\newtheorem{corollary}{Corollary}
\newtheorem{example}{Example}

\def\reff#1{{\rm(\ref{#1})}}

\newcommand\bone{\mathbf 1}
\def\ttheta{\tilde \theta}
\def\tb{\tilde b}
\def\tE{\tilde E}
\newcommand\tW{\tilde W}

\newcommand\EE {\mathbb E}
\newcommand\FF {\mathbb F}

\newcommand\RR {\mathbb R}

\newcommand\PP {\mathbb P}
\newcommand\QQ {\mathbb Q}

\newcommand\bF {\mathbb F}
\newcommand\cA {\mathcal A}

\newcommand\cF {\mathcal F}

\begin{document}

\title[Singular FBSDEs and Emissions Derivatives]{Singular Forward-Backward Stochastic Differential Equations and Emissions Derivatives}

%    author one information
\author{Ren\'e Carmona}
\address{ORFE, Bendheim Center for Finance, Princeton University,
Princeton, NJ  08544, USA.}
\email{rcarmona@princeton.edu}
\thanks{Partially supported  by NSF: DMS-0806591}

\author{Fran\c{c}ois Delarue}
\address{Laboratoire J.-A. Dieudonn\'e,
Universit\'e de Nice Sophia-Antipolis, 
Parc Valrose, 
06108 Cedex 02, Nice, FRANCE}
\email{francois.delarue@unice.fr}

\author{Gilles-Edouard Espinosa}
\address{Centre de Math\'ematiques Appliqu\'ees, Ecole Polytechnique, UMR CNRS 7641,
91128 Palaiseau Cedex, FRANCE}
\email{gilles-edouard.espinosa@polytechnique.org}

\author{Nizar Touzi}
\address{Centre de Math\'ematiques Appliqu\'ees, Ecole Polytechnique, UMR CNRS 7641,
91128 Palaiseau Cedex, FRANCE}
\email{nizar.touzi@polytechnique.edu}
\thanks{The last two authors gratefully acknowledge financial support from the Chair {\it Financial Risks} of the {\it Risk
Foundation} sponsored by Soci\'et\'e G\'en\'erale, the Chair {\it
Derivatives of the Future} sponsored by the {F\'ed\'eration Bancaire
Fran\c{c}aise}, the Chair {\it Finance and Sustainable Development}
sponsored by EDF and CA-CIB.}

\subjclass[2000]{Primary }


\date{July 22, 2010}

\begin{abstract}
We introduce two simple models of forward-backward stochastic differential equations with a singular terminal condition and we explain how and why they appear naturally as models for the valuation of CO${}_2$ emission allowances. Single phase cap-and-trade schemes lead readily to terminal conditions given by indicator functions of the forward component, and using fine partial differential equations estimates, we show that the existence theory of these equations, as well as the properties of the candidates for solution, depend strongly upon the characteristics of  the forward dynamics. Finally, we give a first order Taylor expansion and show how to numerically calibrate some of these models for the purpose of CO${}_2$ option pricing.  
\end{abstract}

\maketitle

\section{\textbf{Introduction}}
\label{se:intro}

This paper is motivated by the mathematical analysis of the emissions markets, as implemented for example in the European Union (EU) Emissions Trading Scheme (ETS).
These market mechanisms have been hailed by some as the most cost efficient way to control Green House Gas (GHG) emissions.
They have been criticized by others  for being a tax in disguise and adding to the burden of industries covered by the regulation.
Implementation of cap-and-trade schemes is not limited to the implementation of the Kyoto protocol. The successful US acid rain program is a case in point. However, a widespread lack of understanding of their properties, and misinformation campaigns by advocacy groups more interested in pushing their political agendas than 
using the results of objective scientific studies have muddied the water and add to the confusion. More mathematical studies are needed to increase the understanding of these market mechanisms and raise the level of awareness of their advantages as well as their shortcomings. This paper was prepared in this spirit.

\vskip 4pt
In a first part, we introduce simple single-firm models inspired by the workings of the electricity markets (electric power generation is responsible for most of the CO${}_2$ emissions worldwide). Despite the specificity of some assumptions, our treatment is quite general in the sense that individual risk averse power producers choose their own utility functions. Moreover, the financial markets in which they trade emission allowances are not assumed to be complete. 

While market incompleteness prevents us from identifying the optimal trading strategy of each producer, we show that, independently of the choice of the utility function, the optimal production or abatement strategy is what we expect by proving mathematically, and in full generality (i.e. without assuming completeness of the markets),  a \emph{folk theorem} in environmental economics: the equilibrium allowance price equals the marginal abatement cost, and market participants implement all the abatement measures whose costs are not greater than the cost of compliance (i.e. the equilibrium price of an allowance).

\vskip 1pt
The next section puts together the economic activities of a large number of producers and searches for the existence of an equilibrium price for the emissions allowances. Such a problem leads naturally to a \emph{forward} stochastic differential equation (SDE) for the aggregate emissions in the economy, and a \emph{backward} stochastic differential equation (BSDE) for the allowance price. However, these equations are "coupled" since a nonlinear function of the \emph{price of carbon} (i.e. the price of an emission allowance) appears in the forward equation giving the dynamics of the aggregate emissions. This feedback of the emission price in the dynamics of the emissions is quite natural. For the purpose of option pricing, this approach was described in \cite{CarmonaHinz} where it was called \emph{detailed risk neutral approach}.

\vskip 1pt
Forward backward stochastic differential equations (FBSDEs) of the type considered in this section have been studied for a long time. See for example \cite{MaYong.book},
or \cite{Pham.book}. However, the FBSDEs we need to consider for the purpose of emission prices have an unusual pecularity: the terminal condition of the backward equation is given by a discontinuous function of the terminal value of the state driven by the forward equation. We use our first model to prove that this lack of continuity is not an issue when the forward dynamics are \emph{strongly elliptic} in the neighborhood of the singularities of the terminal condition, in other words when the volatility of the forward SDE is bounded from below in the neighborhood of the discontinuities of the terminal value. However, using our second equilibrium model, we also show that when the forward dynamics are degenerate (even if they are hypoelliptic), discontinuities in the terminal condition and lack of strong ellipticity in the forward dynamics can conspire to produce point masses in the terminal distribution of the forward component, at the locations of the discontinuities.
This implies that the terminal value of the backward component is not given by a deterministic function of the forward component,  for the forward scenarios ending at the locations of jumps in the terminal condition, and justifies relaxing the definition of a solution of the FBSDE.

Even though we only present a detailed proof for a very specific model for the sake of definiteness, we believe that our result is representative of a large class of models.
Since from the point of view of the definition of "aggregate emissions", the degeneracy of the forward dynamics is expected, this seemingly pathological result should not be overlooked. Indeed, it sheds new light on an absolute continuity assumption made repeatedly in equilibrium analyses, even in discrete time models. See for example \cite{CarmonaFehrHinz} and \cite{CarmonaFehrHinzPorchet}. This assumption was regarded as an annoying technicality, but in the light of the results of this paper, it looks more intrinsic to these types of models. In any case, it fully justifies the need to relax the definition of a solution of a FBSDE when the terminal condition of the backward part jumps. 

\vskip 2pt
A vibrant market for options written on allowance  futures/forward contracts has recently developed and increased in liquidity. See for example \cite{CarmonaHinz} for details on these markets. Reduced formed models have been proposed to price these options. See \cite{CarmonaHinz} or \cite{ChesneyTaschini}. Several attempts have been made at matching the smile (or lack thereof) contained in the quotes published daily by the exchanges. Section \ref{se:option} develops the technology needed to price these options in the context of the equilibrium framework developed in the present paper. We identify the option prices in terms of solutions of nonlinear partial differential equations and we prove
when the dynamics of the aggregate emissions are given by a geometric Brownian motion, a Taylor expansion formula when the nonlinear abatement feedback is small. We derive an explicit integral form for the first order Taylor expansion coefficient which can easily be computed by Monte Carlo methods. We believe that the present paper is the first rigorous attempt to include the nonlinear feedback term in the dynamics of aggregate emissions for the purpose of emissions option pricing.

\vskip 2pt
The final Section \ref{se:option} was motivated by the desire to provide practical tools for the efficient computation of option prices
within the equilibrium framework of the paper. Indeed, because of the nonlinear feedback created by the coupling in the FBSDE, option prices computed from our equilibrium model differ from the \emph{linear} prices computed in \cite{ChesneyTaschini}, \cite{Seifert} and  \cite{CarmonaHinz} in the framework of reduced form models. We derive rigorously an approximation based on the first order asymptotics in the nonlinear feedback. This approximation can be used to compute numerically option prices and has the potential to efficiently fit the implied volatility smile present in recent option price quotes. The final subsection \ref{sub:numerics} illustrates numerically the properties of our approximation.

\vskip 4pt
\emph{Acknowledgements:} We would like to thank two anonymous referees for pointing out inconsistencies in the original proofs of Theorem \ref{thmsec3} and Proposition \ref{prop:1:1}, and for insightful comments which led to improvements in the presentation of the results of the paper.

\section{\textbf{Two Simple Models of Green House Gas Emission Control}}
\label{se:2models}

We first describe the optimization problem of a single power producer facing a carbon  cap-and-trade regulation. We assume that this producer is a small player in the market in the sense that his actions have no impact on prices and that a liquid market for pollution permits exists. In particular, we assume that the price of an allowance is given exogenously, and we use the notation $Y=(Y_t)_{0\le t\le T}$ for the (stochastic) time evolution of the price of such an emission allowance.
For the sake of simplicity we assume that $[0,T]$ is a single phase of the regulation and that no banking or borrowing of the certificates is possible at the end of the phase. For  illustration purposes, we analyze two simple models. Strangely enough, the first steps of these analyses, namely the identifications of the optimal abatement and production strategies, do not require the full force of the sophisticated techniques of optimal stochastic control.

\subsection{\textbf{Modeling First the Emissions Dynamics}}

We assume that the source of randomness in the model is given by $W=(W_t)_{0\le t\le T}$, a finite family of independent  one-dimensional  Wiener processes $W^j=(W^j_t)_{0\le t\le T}$, $1 \leq j \leq d$. In other words, $W_t=(W^1_t,\cdots, W^d_t)$ for each fixed $t\in [0,T]$.
All these Wiener processes are assumed to be defined on a complete probability space $(\Omega,\cF,\PP)$, and we denote by $\FF=\{\cF_t,t\ge 0\}$ the Brownian filtration they generate. Here, $T>0$ is a fixed time horizon representing the end of the regulation period. 

\vskip 2pt
We will eventually extend the model to include $N$ firms, but for the time being, we consider only the problem of one single firm whose production of electricity generates emissions of carbon dioxyde,  and we denote by $E_t$ the cumulative emissions up to time $t$ of the firm. We also denote by $\tE_t$ the perception at time $t$ (for example the conditional expectation) of what the total cumulative emission $E_T$ will be at the end of the time horizon. Clearly, $E$ and $\tE$ can be different stochastic processes, but they have the same terminal values at time $T$, i.e. $E_T=\tE_T$. We will assume that the dynamics of the proxy $\tE$ for the cumulative emissions of the firm are given by an It\^o process of the form: 
 \begin{equation}
 \label{fo:proxyEt}
 \tE_t
 =
 \tE_0+\int_0^t(b_s-\xi_s)ds+\int_0^t\sigma_sdW_s,
 \end{equation}
where $b$ represents the (conditional) expectation of what the rate of emission would be in a world without carbon regulation, in other words in what is usually called \emph{Business As Usual} (BAU for short), while $\xi$ is the instantaneous rate of abatement chosen by the firm. In mathematical terms, $\xi$ represents the control on emission reduction implemented by the firm. Clearly, in such a model, the firm only acts on the drift of its \emph{perceived} emissions. For the sake of simplicity we assume that the processes  $b$ and $\sigma$ are adapted and bounded. Because of the vector nature of the Brownian motion $W$, the volatility process $\sigma$ is in fact a vector of scalar volatility processes $(\sigma^{j})_{1 \leq j \leq d}$. For the purpose of this section, we could use one single scalar Wiener process and one single scalar volatility process as long as we allow the filtration $\bF$ to be larger than the filtration generated by this single Wiener process. 
This fact will be needed when we study a model with more than one firm.

Notice that the formulation \eqref{fo:proxyEt} does not guarantee the positiveness of the perceived emissions process, as one would expect it to be. This issue will be discussed in Proposition \ref{prop:sign} below, where we provide sufficient conditions on the coefficients of \eqref{fo:proxyEt} in order to guarantee the positiveness of the process $\tE$.

\vskip 2pt
Continuing on with the description of the model, we assume that the abatement decision is based on a cost function  $c:\RR\to\RR$ which is assumed to be continuously differentiable ($C^1$ in notation), strictly convex and satisfy Inada-like conditions:
 \begin{equation}
 \label{fo:inada}
 c'(-\infty)=-\infty\qquad\text{and}\qquad c'(+\infty)=+\infty.
 \end{equation}
 Note that $(c')^{-1}$ exists because of the assumption of strict convexity.
Since $c(x)$ can be interpreted as the cost to the firm for an abatement rate of level $x$, without any loss of generality we will also assume  $c(0)=\min c=0$. Notice that (\ref{fo:inada}) implies that $\lim_{x\to\pm\infty}c(x)=+\infty$.

\begin{example}
\label{rem:examples:c}
A typical example of abatement cost function is given by the quadratic cost function $c(x)=\alpha x^2$ for some $\alpha>0$ used in \cite{Seifert}, or more generally the power cost function $c(x)=\alpha|x|^{1+\beta}$ for some $\alpha>0$ and $\beta>0$.
\end{example}

The firm controls its destiny by choosing its own abatement schedule $\xi$ as well as the quantity $\theta$ of pollution permits it holds through trading in the allowance market. For these controls to be admissible, $\xi$ and $\theta$ need only to be progressively measurable processes satisfying the integrability condition 
\begin{equation}
\label{fo:admissibility}
\EE\int_0^T[\theta_t^2+\xi_t^2]dt<\infty.
\end{equation}
We denote by $\cA$ the set of admissible controls $(\xi,\theta)$. Given its initial wealth $x$, the terminal wealth $X_T$ of the firm is given by:
 \begin{equation}
 \label{fo:terminal_wealth}
 X_T=X^{\xi,\theta}_T=x+\int_0^T\theta_tdY_t-\int_0^Tc(\xi_t)dt-E_TY_T.
\end{equation}
The first integral in the right hand side of the above equation gives the proceeds from trading in the allowance market. Recall that we use the notation $Y_t$ for the price of an emission allowance at time $t$. The next term represents the abatement costs, and the last term gives the costs of the emission regulation. Recall also that at this stage, we are not interested in the existence or the formation of this price. We merely assume the existence of a liquid and frictionless market for emission allowances, and that  $Y_t$ is the price at which each firm can buy or sell one allowance at time $t$. The risk preferences of the firm are given by a utility function $U:\RR\to\RR$, which is assumed to be $C^1$, increasing, strictly concave and satisfying the Inada conditions:
 \begin{equation}
 \label{fo:inada4U}
 (U)'(-\infty)=+\infty\qquad\text{and}\qquad(U)'(+\infty)=0.
 \end{equation}
The optimization problem of the firm can be written as the computation of:
 \begin{equation}
 V(x)=\sup_{(\xi,\theta)\in\cA}\EE U(X^{\xi,\theta}_T)
 \end{equation}
where $\EE$ denotes the expectation under the \emph{historical} measure $\PP$, and $\cA$ is the set of abatement and trading  strategies $(\xi,\theta)$ admissible to the firm. The following simple result holds.

\begin{proposition}
The optimal abatement strategy of the firm is given by:
$$
\xi^*_t=[c']^{-1}(Y_t).
$$
\end{proposition}

\begin{remark}
Notice that the optimal abatement schedule is independent of the utility function. The beauty of this simple result is its powerful intuitive meaning: given a price $Y_t$ for an emission allowance, the firm implements all the abatement measures which make sense economically, namely all those costing less than the current market price of one allowance (i.e. one unit of emission).
\end{remark}

\proof
By an immediate integration by parts in the expression (\ref{fo:terminal_wealth}) of the terminal wealth, we see that:
\begin{align*}
\tE_TY_T&=Y_T\left(\tE_0+\int_0^Tb_tdt+\int_0^T\sigma_tdW_t\right)-Y_T\int_0^T\xi_tdt\\
&=Y_T\left(\tE_0+\int_0^Tb_tdt+\int_0^T\sigma_tdW_t\right)-\int_0^TY_t\xi_tdt-\int_0^T\left(\int_0^t\xi_sds\right)dY_t
\end{align*}
so that $X_T=A^{\ttheta}_T+B^\xi_T$ with
$$
A^{\ttheta}_T=\int_0^T\ttheta_tdY_t
              -Y_T\left(\tE_0+\int_0^Tb_tdt+\int_0^T\sigma_tdW_t\right)
$$
where the modified control $\ttheta$ is defined by $\ttheta_t=\theta_t+\int_0^t\xi_sds$, and 
$$
B^\xi_T=x-\int_0^T[c(\xi_t)-Y_t\xi_t]dt.
$$
Notice that $B^\xi$ depends only upon $\xi$ without depending upon $\ttheta$
while $A^{\ttheta}$ depends only upon $\ttheta$ without depending upon $\xi$. The set $\cA$ of admissible controls is equivalently described by varying the couples $(\theta,\xi)$ or $(\ttheta,\xi)$, so when computing the maximum
$$
\sup_{(\theta,\xi)\in\cA}\EE U(X_T)
=\sup_{(\ttheta,\xi)\in\cA}\EE U(A^{\ttheta}_T+B^\xi_T)
$$
one can perform the optimizations over $\ttheta$ and $\xi$ separately, for example by fixing $\ttheta$ and optimizing with respect to $\xi$ before maximizing the result with respect to $\ttheta$. The proof is complete once we notice that $U$ is increasing and that for each $t\in[0,T]$ and each $\omega\in\Omega$, the quantity $B^\xi_T$ is
maximized by the choice $\xi^*_t=(c')^{-1}(Y_t)$.\qed

\begin{remark}
The above result argues neither existence nor uniqueness of an optimal admissible set $(\xi^*,\theta^*)$ of controls. In the context of a complete market, once the optimal rate of abatement $\xi^*$ is implemented, the optimal investment strategy $\theta^*$ should hedge the financial risk created by the 
implementation of the abatement strategy. This fact can be proved using the classical tools of portfolio optimization in the case of complete market models.
Indeed, if we introduce the convex dual $\tilde U$ of $U$ defined by:
$$
\tilde U(y):=\sup_{x}\{U(x)-xy\}
$$
and the function $I$ by  $I=(U')^{-1}$ so that $\tilde U(y)=U\circ I(y)-yI(y)$ and if we denote by $\EE$ and $\EE^\QQ$ respectively the expectations with respect to $\PP$ and 
the unique equivalent measure $\QQ$ under which $Y$ is a martingale (we write $Z_t$ for  its volatility given by the martingale representation theorem), then from the a.s. inequality
$$
U(X^{\xi,\theta}_T)-y\frac{d\QQ}{d\PP}X^{\xi,\theta}_T\leq U\circ I\left(y\frac{d\QQ}{d\PP}\right)-y\frac{d\QQ}{d\PP}I\left(y\frac{d\QQ}{d\PP}\right),
$$
valid for  any admissible $(\xi,\theta)$, and $y\in\RR$, we get
$$
\EE U(X^{\xi,\theta}_T)\leq \EE U\circ I\left(y\frac{d\QQ}{d\PP}\right)+y\EE^\QQ\left[X^{\xi,\theta}_T-I\left(y\frac{d\QQ}{d\PP}\right)\right]
$$
after taking expectations under $\PP$. Computing $\EE^\QQ X^{\xi,\theta}_T$ by integration by parts we get:
\begin{align*}
\EE U(X^{\xi,\theta}_T)\leq \EE U\circ I\left(y\frac{d\QQ}{d\PP}\right)+y\left[x-\EE^\QQ\int_0^T[c\circ (c')^{-1}(Y_t)+Y_t(b_t-\right.&\left.(c')^{-1}(Y_t))+\sigma_tZ_t]dt\right.\\
&\left.-\EE^\QQ I\left(y\frac{d\QQ}{d\PP}\right)\right]
\end{align*}
if we use the optimal rate of abatement. So if we choose $y=\hat y\in\RR$ as the unique solution of:
$$
\EE^\QQ I\left(\hat y\frac{d\QQ}{d\PP}\right)=x-\EE^\QQ\int_0^Tc\circ (c')^{-1}(Y_t)+Y_t(b_t-(c')^{-1}(Y_t))+\sigma_tZ_tdt.
$$
it follows that
$$
\EE^\QQ X^{\hat\xi,\theta}_T=\EE^\QQ I\left(\hat y\frac{d\QQ}{d\PP}\right),
$$
and finally, if the market is complete, the claim $I \left(\hat y \frac{d\QQ}{d\PP}\right)$ is attainable by a certain $\theta^*$. This completes the proof.
\end{remark}

\subsection{\textbf{Modeling the Electricity Price First}}
\label{sub:elecprice}
We consider a second model for which again, part of the global stochastic optimization problem reduces to a mere path-by-path optimization. As before, the model is simplistic, especially in the case of a single firm in a regulatory environment with a liquid frictionless market for emission allowances. However, this model will become very informative later on when we consider $N$ firms interacting on the same market, and we try to construct the allowance price $Y_t$ by solving a Forward-Backward Stochastic Differential Equation (FBSDE).
The model concerns an economy with one production good (say electricity) whose production is the source of a negative externality (say GHG emissions). Its price $(P_t)_{0\le t\le T}$ evolves according to the following It\^o stochastic differential equation:
\begin{equation}
\label{fo:poft}
dP_t=\mu(P_t)dt+\sigma(P_t)dW_t
\end{equation}
where the deterministic functions $\mu$ and $\sigma$ are assumed to be $C^1$ with bounded derivatives.
At each time $t\in[0,T]$, the firm chooses its instantaneous rate of production $q_t$ and its production costs are $c(q_t)$ where $c$ is a function $c:\RR_+\hookrightarrow \RR$ which is assumed to be $C^1$ and strictly convex.
With these notations, the profits and losses from the production  at the end of the period $[0,T]$, are given by the integral:
$$
 \int_0^T[P_tq_t-c(q_t)]dt.
$$
The emission regulation mandates that at the end of the period $[0,T]$, the cumulative emissions of each firm be measured, and that one emission permit be redeemed per unit of emission. As before, we denote by $(Y_t)_{0\le t\le T}$ the process giving the price of one emission allowance. For the sake of simplicity, we assume that the cumulative emissions $E_t$ up to time $t$ are proportional to the production in the sense that $E_t=\epsilon Q_t$ where the positive number $\epsilon$ represents the rate of emission of the production technology used by the firm, and $Q_t$ denotes the cumulative production up to and including time $t$:
$$
Q_t=\int_0^tq_sds.
$$
At the end of the time horizon, the cost incurred by the firm because of the regulation is given by $E_TY_T=\epsilon Q_TY_T$.
The firm may purchase allowances: we denote by $\theta_t$ the amount of allowances held by the firm at time $t$. Under these conditions, the terminal wealth of the firm is given by:
 \begin{equation}
 X_T=X^{q,\theta}_T=x+\int_0^T\theta_tdY_t+\int_0^T[P_tq_t-c(q_t)]dt-\epsilon Q_TY_T
 \end{equation}
where as before, we used the notation $x$ for the initial wealth of the firm. The first integral in the right hand side of the above equation gives the proceeds from trading in the allowance market, the next term gives the profits from the production and the sale of electricity, and the last term gives the costs of the emission regulation. We assume that the risk preferences of the firm are given by a utility function $U:\RR\to\RR$, which is assumed to be $C^1$, increasing, strictly concave and satisfying the Inada conditions \reff{fo:inada4U} stated earlier. As before, the optimization problem of the firm can be written as:
 \begin{equation}
 V(x)=\sup_{(q,\theta)\in\cA}\EE U(X^{q,\theta}_T)
 \end{equation}
where $\EE$ denotes the expectation under the \emph{historical} measure $\PP$, and $\cA$ is the set of admissible production and trading  strategies $(q,\theta)$. This problem  is similar to those studied in \cite{belaouar-fahim-touzi} where the equilibrium issue is not addressed. As before, for these controls to be admissible, $q$ and $\theta$ need only be adapted processes satisfying the integrability condition 
\begin{equation}
\label{fo:qadmissibility}
\EE\int_0^T[\theta_t^2+q_t^2]dt<\infty.
\end{equation}

\begin{proposition}
The optimal production strategy of the firm is given by:
$$
q^*_t=(c')^{-1}(P_t-\epsilon Y_t).
$$
\end{proposition}

\begin{remark}
As before, the optimal production strategy $q^*$ is independent of the risk aversion (i.e. the utility function) of the firm.
The intuitive interpretation of this result is clear: once a firm observes both prices $P_t$ and $Y_t$, it computes the price for which it can sell the good minus the price it will have to pay because of the emission regulation, and the firm uses this corrected price to choose its optimal rate of production in the usual way.
\end{remark}

\proof
A simple integration by part (notice that $E_t$ is of bounded variations) gives:
\begin{equation}
Q_TY_T=\int_0^TY_tdQ_t+\int_0^TQ_tdY_t=\int_0^TY_tq_tdt+\int_0^TQ_tdY_t,
\end{equation}
so that $X_T=A^{\ttheta}_T+B^q_T$ with
$$
A^{\ttheta}_T=\int_0^T\ttheta_tdY_t\qquad\text{with}\qquad \ttheta_t=\theta_t-\epsilon\int_0^tq_sds
$$
which depends only upon $\ttheta$ and 
$$
B^q_T=x+\int_0^T[(P_t-\epsilon Y_t)q_t-c(q_t)]dt,
$$
which depends only upon $q$ without depending upon $\ttheta$. 
Since the set $\cA$ of admissible controls is equivalently described by varying the couples $(q,\theta)$ or $(q,\ttheta)$, when computing the maximum
$$
\sup_{(q,\theta)\in\cA}\EE\{U(X_T)\}
=\sup_{(q,\ttheta)\in\cA}\EE\{U(A^{\ttheta}_T+B^q_T)\}
$$
one can perform the optimizations over $q$ and $\ttheta$ separately, for example by fixing $\ttheta$ and optimizing with respect to $q$ before maximizing the result with respect to $\ttheta$. The proof is complete once we notice that $U$ is increasing and that for each $t\in[0,T]$ and each $\omega\in\Omega$, the quantity $B^q_T$ is
maximized by the choice $q^*_t=(c')^{-1}(P_t-\epsilon Y_t)$. \qed

\section{\textbf{Allowance Equilibrium Price and a First Singular  FBSDE}}
\label{se:equilibrium}

The goal of this section is to extend the first model introduced in section \ref{se:2models} to an economy with $N$ firms, and solve for the allowance price.

\subsection{Switching to a Risk Neutral Framework}
As before, we assume that $Y=(Y_t)_{t\in[0,T]}$ is the price of one allowance in a one-compliance period cap-and-trade model,
and that the market for allowances is frictionless and liquid. In the absence of arbitrage, $Y$ is a martingale for a measure $\QQ$ equivalent to the historical measure $\PP$. Because we are in a Brownian filtration, 
$$
\frac{d\QQ}{d\PP}=\exp\left[\int_0^T\alpha_t dW_t\,-\,\frac12\int_0^T|\alpha_t|^2dt
\right]
$$
for some sequence $\alpha=(\alpha_t)_{t\in[0,T]}$ of adapted processes. By Girsanov's theorem, the process $\tW=(\tW_t)_{t\in[0,T]}$ defined by
$$
\tW_t=W_t-\int_0^t\alpha_sds
$$
is a Wiener process for $\QQ$ so that equation \reff{fo:proxyEt} giving the dynamics of the perceived emissions of a firm now reads:
$$
d\tE_t=(\tb_t-\xi_t)dt\,+\sigma_t d\tW_t
$$
under $\QQ$, where the new drift $\tb$ is defined by $\tb_t=b_t+\sigma_t\alpha_t$ for all $t\in[0,T]$.

\subsection{Market Model with $N$ Firms}
We now consider an economy comprising $N$ firms labelled by $\{1,\cdots,N\}$, and we work in the risk neutral framework for allowance trading discussed above. When a specific quantity such as cost function, utility, cumulative emission, trading strategy, $\ldots$ depends upon a firm, we
use a superscript ${}^i$ to emphasize the dependence upon the $i$-th firm.
So in equilibrium (i.e. whenever each firm implements its optimal abatement strategy), for each firm $i\in\{1,\cdots,N\}$ we have
$$
d\tE^i_t=\{\tb^i_t-[(c^i)']^{-1}(Y_t)\}dt+\sigma^i_t d\tW_t
$$
with given initial perceived emissions $\tE^i_0$. Consequently, the aggregate perceived emission ${\tE}$ defined by
$$
{\tE}_t=\sum_{i=1}^N\tE^i_t
$$
satisfies
$$
d{\tE}_t=\big(b_t-f(Y_t)\big)dt +{\sigma}_t d\tW_t,
$$
where
$$
b_t=\sum_{i=1}^N \tb^i_t,\qquad {\sigma}_t=\sum_{i=1}^N\sigma^i_t\qquad\text{and}\qquad f(x)=\sum_{i=1}^N[(c^i)']^{-1}(x).
$$
Again, since we are in a Brownian filtration, it follows from the martingale representation theorem that there exists a progressively measurable process $Z=(Z_t)_{t\in[0,T]}$ such that
\begin{eqnarray*}
dY_t=Z_t d\tW_t
&\mbox{and}&
\EE^{\mathbb Q}\int_0^T |Z_t|^2dt<\infty.
\end{eqnarray*}
Furthermore, in order to entertain a concrete existence and uniqueness result, we assume that $\tW$ is one-dimensional and that there exist deterministic continuous functions $[0,T]\times\RR\ni (t,e)\hookrightarrow b(t,e) \in \RR$ and $[0,T] \times \RR \ni t\hookrightarrow \sigma(t,e) \in \RR$ such that $b_t=b(t,\tE_t)$ and $\sigma_t=\sigma(t,\tE_{t})$, for all $t\in[0,T]$, $\QQ$-a.s.

\vskip 4pt\noindent
Consequently, the processes $\tE$, $Y$, and $Z$ satisfy a system of Forward Backward Stochastic Differential Equations (FBSDEs for short) under ${\mathbb Q}$, which we restate for the sake of later reference:
\begin{equation}
\label{fo:RNfbsde}
\begin{cases}
&d\tE_t=\big(b(t,\tE_t)-f(Y_t)\big)dt+\sigma(t,\tE_{t})d\tW_t,~\quad \mbox{with given}~~\tE_0\in\RR\\
&dY_t=Z_td\tW_t,~\quad Y_T=\lambda \bone_{[\Lambda,+\infty)}(\tE_T).
\end{cases}
\end{equation}
The fact that the terminal condition for $Y_T$ is given by an indicator function results from the equilibrium analysis of these markets. See \cite{CarmonaFehrHinz} and \cite{CarmonaFehrHinzPorchet}. $\Lambda$ is the global emission target set by the regulator for the entire economy. It represents the \emph{cap} part of the cap-and-trade scheme. $\lambda$ is the penalty that firms have to pay for each emission unit not covered by the redemption of an allowance. Currently, this penalty is $100$ euros in the European Union Emission Trading Scheme (EU ETS).
Notice that since all the cost functions $c^i$ are strictly convex, $f$ is strictly increasing. We shall make the following additional assumptions:
 \begin{eqnarray}
 &b(t,e) \mbox{ and } \sigma(t,e) ~~\mbox{are Lipschitz in $e$ uniformly in $t$}, \hspace{160pt}&
 \label{b-Niz}
 \\
 &\mbox{there exists an open ball $U \subset \RR^2$, $U \ni (T,\Lambda)$, such that} 
\inf_{(t,e) \in U \cap [0,T] \times \RR} \sigma^2(t,e) >0,  &
  \label{sigma-Niz} 
 \\
 &f~~\mbox{is Lipschitz continuous (and strictly increasing). \hspace{156pt} }&
 \label{f-Niz}
 \end{eqnarray}
We denote by $\mathbb{H}^0$ the collection of all $\RR$-valued progressively measurable processes on $[0,T]\times\RR$, and we introduce the subsets:
 \begin{eqnarray*}
 \mathbb{H}^2
 :=
 \Big\{Z\in\mathbb{H}^0;\;\EE^{\mathbb Q}\int_0^T|Z_s|^2ds<\infty\Big\}
 &\mbox{and}&
 \mathbb{S}^2
 :=
 \Big\{Y\in\mathbb{H}^0;\;\EE^{\mathbb Q}[\sup_{0\le t\le T}|Y_s|^2]<\infty\Big\}.
 \end{eqnarray*}

\subsection{Solving the Singular Equilibrium FBSDE}
\label{subsec:solving}
The purpose of this subsection is to prove existence and uniqueness of a solution to FBSDE \reff{fo:RNfbsde}.

\begin{theorem}
\label{thmsec3}
If assumptions \reff{b-Niz} to \reff{f-Niz} hold for a given $\Lambda \in \RR$, then, for any $\lambda>0$, FBSDE \reff{fo:RNfbsde} admits a unique solution $(\tE,Y,Z)$ $\in$ $\mathbb{S}^2\times\mathbb{S}^2\times\mathbb{H}^2$. Moreover, for any $t \in [0,T]$, $\tE_t$ is non-increasing with respect to $\lambda$ and non-decreasing with respect to $\Lambda$.
\end{theorem}

\begin{proof}
For any function $\varphi: \RR \hookrightarrow \RR$, we write FBSDE($\varphi$) for the FBSDE  \reff{fo:RNfbsde} when the function $g=\lambda {\mathbf 1}_{[\Lambda,+\infty)}$ appearing in the terminal condition in the backward component of \reff{fo:RNfbsde} is replaced by $\varphi$.

(i) We first prove uniqueness. Let $(\tE,Y,Z)$ and $(\tE',Y',Z')$ be two solutions of FBSDE \reff{fo:RNfbsde}. Clearly it is sufficient to prove that $Y=Y'$. Let us set:
\begin{equation*}
\begin{split}
&\delta E_t:=\tE_t-\tE'_t,
 \quad\delta Y_t:=Y_t-Y'_t,
 \quad\delta Z_t:=Z_t-Z'_t,
\\ 
&\beta_t:=\frac{b(t,\tE_t)-b(t,\tE'_t)}{\delta E_t}\mathbf{1}_{\{\delta E_t\neq 0\}}, \quad \Sigma_{t}:= \frac{\sigma(t,\tE_t)-\sigma(t,\tE'_t)}{\delta E_t}\mathbf{1}_{\{\delta E_t\neq 0\}}. 
\end{split}
\end{equation*}
Notice that $(\beta_t)_{0\le t\le T}$ and $(\Sigma_{t})_{0 \le t \le T}$ are bounded processes. By direct calculation, we see that
 \begin{equation*}
 d(B_t\delta E_t\delta Y_t)
 =
 -B_t\delta Y_t\big(f(Y_t)-f(Y'_t)\big)dt
 +B_t\delta E_t\delta Z_t d\tW_t,
 \end{equation*}
 where
 \begin{equation*}
 B_t:= \exp \biggl( \int_0^t \bigl(\frac{\Sigma_{s}^2}{2}-\beta_s \bigr)ds-\int_{0}^t \Sigma_{s}d\tW_{s} \biggr).
 \end{equation*}
Since $\delta E_0=0$ and $\delta E_T\delta Y_T=(\tE_T-\tE'_T)\big(g(\tE_T)-g(\tE'_T)\big)\ge 0$, because $g$ is  nondecreasing, this implies that
 $$
 \EE^{\mathbb Q}\biggl[\int_0^T B_t\delta Y_t\big(f(Y_t)-f(Y'_t)\big)dt\biggr]
 \;\le\;
 0.
 $$
Since $B_t>0$ and $f$ is (strictly) increasing, this implies that $\delta Y=0$ $dt\otimes d\QQ-$a.e. and therefore $Y=Y'$ by continuity. 

(ii) We next prove existence. Let $(g^n)_{n\ge 1}$ be an increasing sequence of smooth non-decreasing functions with $g^n\in[0,\lambda]$ and such that $g^n\longrightarrow g^-= \lambda \mathbf{1}_{(\Lambda,\infty)}$.

\noindent (ii-1) We first prove the existence of a solution when the boundary condition is given by $g^n$.
For every $n \geq 1$, the FBSDE($g^n$) satisfies the assumption of Theorems 5.6 and 7.1 in \cite{MaWuZhangZhang}
with $b_{3}=0$, $f_{1}=f_{2}=f_{3}=0$, $\sigma_{2}=\sigma_{3}=0$, $b_{2} \leq 0$
(by  \reff{f-Niz}) and  $\underline{h}=0$ (since $g^n$ is non-decreasing)
so that Condition (5.11)  in \cite{MaWuZhangZhang} holds with $\lambda=0$ and $\underline{F}(t,0)=0$ for any $\varepsilon >0$. By Theorem 7.1 in \cite{MaWuZhangZhang}, the FBSDE($g^n$) has a unique solution $(\tE^{n},Y^{n},Z^{n})$ $\in$ $\mathbb{S}^2\times\mathbb{S}^2\times\mathbb{H}^2$. Moreover, it holds $Y^{n}_t=u^{n}(t,\tE^{n}_t)$, $0 \leq t \leq T$, for some deterministic function $u^{n}$. In contrast with \cite{MaWuZhangZhang}, the function $u^n$ is not a random field but a deterministic function since the coefficients of the FBSDE are deterministic. We refer to 
\cite{PardouxPeng} for the general construction of $u^n$ when the coefficients are deterministic.
Since the sequence $(g^n)_{n \geq 1}$ is increasing we deduce from the comparison principle
\cite[Theorem 8.6]{MaWuZhangZhang}, which applies under the same assumption as \cite[Theorem 7.1]{MaWuZhangZhang},
  that, for any $t \in [0,T]$, the sequence of functions $(u^n(t,\cdot))_{n\ge 1}$ is non-decreasing. By \cite[Theorem 8.6]{MaWuZhangZhang} again, $u^n$ is non-decreasing in $\lambda$ and non-increasing in $\Lambda$. Since $g^n$ is $[0,\lambda]$-valued and $u^n(t,e)={\mathbb E}^{\mathbb Q}[g^n(\tE^n_{T}) \vert \tE^n_{t} =e ]$, we deduce that $u^n$ is $[0,\lambda]$-valued as well. Since the sequence of functions $(u^n)_{n\ge 1}$ is non-decreasing, we may then define:
 $$
 u(t,e):=\lim_{n\to\infty}\uparrow u^n(t,e),\qquad t\in[0,T],~e\in\RR.
 $$
Clearly, $u$ is $[0,\lambda]$-valued and $u(t,\cdot)$ is a non-decreasing function for any $t \in [0,T]$. Moreover, $u$ is non-decreasing in $\lambda$ and non-increasing in $\Lambda$. 

By Theorem 6.1--$(iii)$ and Theorem 7.1--$(i)$ in \cite{MaWuZhangZhang}, we know that, for every $n \geq 1$, the function 
$u^n$ is Lipschitz continuous with respect to $e$, uniformly in $t \in [0,T]$. Actually, we claim that, for any $\delta \in (0,T)$, the function $u^n(t,\cdot)$ is Lipschitz continuous in $e$, uniformly in 
$t \in [0,T-\delta]$ and in $n \geq 1$. The proof follows again from Theorem 6.1--$(iii)$ and Theorem 7.1--$(i)$ in \cite{MaWuZhangZhang}. To be more specific, we need to establish a uniform upper bound for the bounded solutions $\bar{y}$ to the first ODE in
\cite[(3.12)]{MaWuZhangZhang} associated with an arbitrary positive terminal condition $\bar{y}_{T}= \bar{h}>0$. Namely, for given bounded (measurable) functions $b_{1} : [0,T] \ni t \hookrightarrow b_{1}(t) \in \RR_{+}$ and
$b_{2} : [0,T] \ni t \hookrightarrow b_{2}(t) \in \RR_{+}$, with $\inf_{t \in [0,T]} 
b_{2}(t) >0$, we are seeking an upper bound for any bounded $(\bar{y}_{t})_{0 \leq t \leq T}$ satisfying 
\begin{equation*}
\bar{y}_{t} = \bar{y}_{T} + \int_{t}^T \bigl( b_{1}(s) \vert \bar{y}_{s} \vert - b_{2}(s) \bar{y}_{s}^2 \bigr) ds \quad ; \quad 
\bar{y}_{T} = \bar{h}> 0.
\end{equation*}
Here $b_{1}(t)$ is understood as an upper bound for the derivative of $b$ with respect to $x$, and 
$b_{2}$ as a lower bound for the derivative of $f$ with respect to $y$. As long as $\bar{y}_{t}$ doesn't vanish, we deduce from a simple computation that
\begin{equation*}
\bar{y}_{t} = \exp \biggl(  \int_{t}^T b_{1}(s) ds \biggr) \biggl( \frac{1}{\bar{y}_{T}} + \int_{t}^T b_{2}(s) \exp \biggl( \int_{s}^T b_{1}(r) dr \biggr) ds \biggr)^{-1}.
\end{equation*}
Since the right-hand side above is always (strictly) positive, we conclude that it is indeed a solution for any $t \in [0,T]$. Therefore, there exists a constant $C$, independent of $\bar{y}_{T}$, such that 
$\bar{y}_{t} \leq C/(T-t)$ for any $t \in [0,T)$. By $(iii)$ in Theorem 6.1 and $(i)$ in Theorem 7.1 in \cite{MaWuZhangZhang}, we deduce that, for any $\delta \in (0,T]$, the function $u^n(t,\cdot)$ is Lipschitz continuous with respect to $e$, uniformly in $t \in [0,T-\delta]$ and $n \geq 1$. Letting $n$ tend to $+ \infty$, we deduce that the same holds for $u$. 

Notice that the process $\tE^n$ solves the (forward) stochastic differential equation
 \begin{equation*}
 d\tE^n_t
 =
 \big(b(t,\tE^n_t)-f\circ u^n(t,\tE^n_t)\big)dt
 +\sigma(t,\tE_{t}^n)d\tW_t, \quad t \in [0,T),
 \end{equation*}
where here and in the following, we use the notation $f\circ u$ for the composition of the functions $f$ and $u$.
Since $f$ is increasing and the sequence $(u^n)_{n\ge 1}$ is non-decreasing, it follows from the comparison theorem for (forward) stochastic differential equations that the sequence of processes $(\tE^n)_{n\ge 1}$ is non-increasing. We may then define:
 $$
 \hat E_t:=\lim_{n\to\infty}\downarrow \tE^n_t
\qquad\mbox{for}~~t\in [0,T].
 $$

\noindent (ii-2) To identify the dynamics of the limiting process $\hat E$, we introduce the process $\tE$ defined on $[0,T)$ as the unique strong solution of the stochastic differential equation
 $$
 d\tE_t
 =
 (b-f\circ u)(t,\tE_t)dt + \sigma(t,\tE_{t})d\tW_t, \quad t \in [0,T) \ ; \ \tE_{0} = 0.
 $$
The fact that the function $u$ is bounded and Lipschitz-continuous in space (locally in time), together with our assumptions on $b$, $f$ and $\sigma$ guarantee the existence and uniqueness of such a strong solution. Since $b$ is at most of linear growth and $u$ is bounded, the solution cannot explode as $t$ tend to $T$, so that the process $(\tE_{t})_{0 \leq t < T}$ can be extended by continuity to the closed interval $[0,T]$. 
Since $u$ is Lipschitz continuous with respect to $e$, uniformly in $t \in [0,T-\delta]$ for any $\delta \in (0,T)$, we deduce from the classical comparison result for stochastic differential equations that $\tE^n_{t} \geq \tE_{t}$ for any $t \in [0,T)$. Letting $t$ tend to $T$, it also holds
$\tE^n_{T} \geq \tE_{T}$. Since, for any $n \geq 1$, $u^n(t,e)= \EE^{\mathbb Q} [g^n(\tE^n_T)|\tE^n_t=e]$, for $(t,e) \in [0,T) \times \RR$, and $g^n$ is a non-decreasing function, we deduce that $u^n(t,.)$ is a non-decreasing function as well. Obviously, the same holds for $u(t,\cdot)$. We then use the fact that $\tE^n\ge \tE$ together with the increase of $u^n(t,.)$ to compute, using It\^o's formula, that, for any $t \in [0,T]$:
 \begin{equation}
 \label{eq:2:12:1}
 \begin{split}
 (\tE^n_t-\tE_t)^2
 &=
 2  \int_0^t (\tE^n_s-\tE_s)\big((b-f\circ u^n)(s,\tE^n_s)
               -(b-f\circ u)(s,\tE_s)
          \big)ds 
          \\
          &\hspace{10pt}  + \int_{0}^t \vert \sigma(s,\tE^n_{s}) - \sigma(s,\tE_{s}) \vert^2 ds
          + 2 \int_{0}^t (\tE^n_s-\tE_s) \bigl( \sigma(s,\tE^n_{s}) - \sigma(s,\tE_{s}) \bigr) d\tW_{s}
 \\
 &\le
 C   \int_0^t \big|\tE^n_s-\tE_s\big|^2ds+ 
 2   \int_0^t (\tE^n_s-\tE_s)(f\circ u-f\circ u^n)(s,\tE_s)ds
          \\
  &\hspace{10pt} + 2 \int_{0}^t (\tE^n_s-\tE_s) \bigl( \sigma(s,\tE^n_{s}) - \sigma(s,\tE_{s}) \bigr) d\tW_{s}
 \\
 &\leq
 (C+1)  \int_0^t \big|\tE^n_s-\tE_s\big|^2ds 
 +  \int_0^t \big|(f\circ u-f\circ u^n)(s,\tE_s)\big|^2ds
          \\
          &\hspace{10pt}   + 2 \int_{0}^t (\tE^n_s-\tE_s) \bigl( \sigma(s,\tE^n_{s}) - \sigma(s,\tE_{s}) \bigr) d\tW_{s},
\end{split}
 \end{equation}
by the Lipschitz property of the coefficients $b$ and $\sigma$.  Taking expectation, we deduce
 \begin{equation*}
\EE^{\mathbb Q} \bigl[(\tE^n_t-\tE_t)^2\bigr]
\leq  (C+1) \EE^{\mathbb Q} \int_0^t \big|\tE^n_s-\tE_s\big|^2ds 
 + \EE^{\mathbb Q} \int_0^t \big|(f\circ u-f\circ u^n)(s,\tE_s)\big|^2ds.
 \end{equation*}
Then
 \begin{eqnarray*}
 \EE^{\mathbb Q} \bigl[ (\tE^n_t-\tE_t)^2 \bigr]
 &\le&
 (C+1)\int_0^t  \EE^{\mathbb Q} \bigl[(\tE^n_s-\tE_s)^2\bigr]ds 
 +\varepsilon^n
 \end{eqnarray*}
where $\varepsilon^n
 :=
 \EE^{\mathbb Q}\big[\int_0^T\big|(f\circ u-f\circ u^n)(s,\tE_s)\big|^2ds\big]\longrightarrow 0$, by the dominated convergence theorem. Therefore it follows from Gronwall's inequality that $\sup_{0 \leq t \leq T}\EE^{\mathbb Q} [(\tE^n_t-\tE_t)^2] \rightarrow 0$ as $n$ tends to $+ \infty$. Repeating the argument, but using in addition the Burkh\"older--Davis--Gundy inequality in \eqref{eq:2:12:1}, we deduce that 
 $\tE^n \longrightarrow \tE$ in $\mathbb{S}^2$, and as a consequence,  $\hat E=\tE$.

\vskip 2pt\noindent
(ii-3) The key point to pass to the limit in the backward equation is to prove that $\QQ[\tE_T= \Lambda] =0$. Given a small real $\delta >0$, we write 
\begin{equation}
\label{eq:29:11:1}
\begin{split}
\QQ[\tE_T= \Lambda] &=  \QQ\bigl[\tE_T= \Lambda, (t,\tE_{t})_{T-\delta \leq t \leq T} \in U\bigr]
\\
&\hspace{15pt}+ \QQ\bigl[  \tE_T= \Lambda, \exists t \in [T-\delta,T] : (t,\tE_{t}) \not \in U\bigr],
\end{split}
\end{equation}
where $U$ is as in \eqref{sigma-Niz}. (Here, the notation $(t,\tE_{t})_{T-\delta \leq t \leq T} \in U$ means that 
$(t,\tE_{t}) \in U$ for any $t \in [T-\delta,T]$.) On the event $\{ (t,\tE_{t})_{T-\delta \leq t \leq T} \in U\}$, the process 
$(\tE_{t})_{T-\delta \leq t \leq T}$ 
coincides with $(X_{t})_{T-\delta \leq t \leq T}$, solution to 
\begin{equation*}
X_{t} = \tE_{T-\delta} + \int_{T-\delta}^t \bigl( b(s,X_{s}) - f \circ u (s,X_{s}) \bigr) + \int_{T-\delta}^t
\tilde{\sigma}(s,X_{s}) d\tilde{W}_{s}, \quad T-\delta \leq t \leq T, 
\end{equation*}
where $\tilde{\sigma} : [0,T] \times \RR \hookrightarrow \RR$ is a given bounded and continuous function which is Lipschitz continuous with respect to $e$, which satisfies $\inf_{[0,T] \times \RR} \tilde{\sigma} >0$, and which coincides with $\sigma$ on $U$.  

Since $\tilde{\sigma}^{-1}$ is bounded and $f$ is bounded on $[0,\lambda]$, we may introduce an equivalent measure $\tilde\QQ\sim\QQ$ under which the process $\tilde{B}_t:=\tW_t-\tilde{\sigma}^{-1}(t,X_{t}) (f\circ u)(t,X_t)$, $t\in[T-\delta,T]$, is a Brownian motion. Then $X$ solves the stochastic differential equation
 \begin{equation}\label{EnE Niz}
 dX_t=b(t,X_t)dt+\tilde{\sigma}(t,X_{t})d\tilde{B}_t, \quad t \in [T-\delta,T] \ ; \ X_{T-\delta}=\tE_{T-\delta}.
 \end{equation}
By Theorem 2.3.1 in \cite{Nualart}, the conditional law,  under $\tilde\QQ$, of $X_T$ given the initial condition $X_{T-\delta}$ has a density with respect to the Lebesgue measure. Consequently,  $\tilde\QQ[X_T= \Lambda] =0$, and the same holds true under the equivalent measure $\QQ$. Therefore, 
 \begin{equation*}
%\label{fo:noPointMass}
 \QQ\bigl[\tE_T= \Lambda, (t,\tE_{t})_{T-\delta \leq t \leq T} \in U\bigr] = 0.
 \end{equation*}
By \eqref{eq:29:11:1}, we deduce
\begin{equation*}
\begin{split}
\QQ[\tE_T= \Lambda] &= \QQ\bigl[ \tE_{T}=\Lambda, \exists t \in [T-\delta,T] : (t,\tE_{t}) \not \in U\bigr]
\\
&\leq \QQ\bigl[ \sup_{T-\delta \leq s \leq T} \vert (s,\tE_{s}) - (T,\tE_{T}) \vert 
\geq {\rm dist}((T,\Lambda),U^{\complement}) \bigr].
\end{split}
\end{equation*}
As $\delta$ tends to 0, the right-hand side above tends to 0, so that 
 \begin{equation}
 \label{fo:noPointMass}
 \QQ[\tE_T= \Lambda] =0,
 \end{equation}
 which implies that we can use $g^-=\lambda \mathbf{1}_{(\Lambda,\infty)}$ instead of $g=\lambda \mathbf{1}_{[\Lambda,\infty)}$
 in \reff{fo:RNfbsde}. Moreover, we also have:
 \begin{equation}\label{ac Niz}
\lim_{n\to\infty} \QQ[\tE^n_T> \Lambda|\cF_t]= \QQ[\tE_T> \Lambda|\cF_t]
\end{equation}
for each $t<T$. The fact that $g^n\le g$ implies:
 $$
 Y^n_t= \EE_t^{\mathbb Q}[g^n(\tE^n_T)]\le \EE_t^{\mathbb Q}[g(\tE^n_T)]\longrightarrow  \EE_t^{\mathbb Q}[g(\tE_T)]
 $$
as $n\to \infty$ by (\ref{ac Niz}). On the other hand, since $\tE^n_T\ge \tE_T$, it follows from the non-decrease of $g_n$, the dominated convergence theorem, and (\ref{ac Niz}) that
 $$
 Y^n_t=\EE_t^{\mathbb Q}[g^n(\tE^n_T)]\ge\EE_t^{\mathbb Q}[g^n(\tE_T)]\longrightarrow \EE_t^{\mathbb Q}[g(\tE_T)].
 $$
Hence $Y^n_t\longrightarrow Y_t:= \EE_t^{\mathbb Q}[g(\tE_T)]$. Now, let $Z\in\mathbb{H}^2$ be such that
 $$
 Y_t=g(\tE_T)-\int_t^T Z_sd\tW_s, \qquad t\in[0,T].
 $$
Notice that $Y$ takes values in $[0,\lambda]$, and therefore $Y\in\mathbb{S}^2$. Similarly, using the increase and the decrease of the sequences $(u^n)_{n \geq 1}$ and $(E^n)_{n \geq 1}$ respectively, together with the increase of the functions $u^n(t,.)$ and $u(t,.)$ and the continuity of the function $u(t,\cdot)$ for $t \in [0,T)$, we see that for $t\in[0,T)$:
 $$
 u(t,\tE_t)
 =
 \lim_{n\to\infty} u^n(t,\tE_t)
 \le \liminf_{n\to\infty} u^n(t,\tE^n_t)
 \le \limsup_{n\to\infty} u^n(t,\tE^n_t)
 \le \lim_{n\to\infty} u(t,\tE^n_t)
 =u(t,\tE_t).
 $$
Since $Y^n_t=u^n(t,\tE^n_t)$, this shows that $Y_t=u(t,\tE_t)$ on $[0,T)$, and the proof of existence of a solution is complete.
\end{proof}

\vskip 2pt\noindent
%(iii) It finally remains to check that, for any $t \in [0,T]$, $\tE_{t}$ is non-increasing with respect to $\lambda$ and non-decreasing with respect to 
%$\Lambda$. Because of the comparison principle for forward SDEs and of the decrease of $f$, it is sufficient to prove that 
%the limit value function $u$ is non-decreasing with respect to $\lambda$ and non-increasing with respect to $\Lambda$. By the approximation %argument we used in point (ii-1), it is sufficient to prove that, for each $m,n \geq 1$, $u^{m,n}$ is non-decreasing with respect to $\lambda$ %and 
%non-increasing with respect to $\Lambda$. This may be shown by a similar argument as the one in (ii-1), by reducing the fully coupled %forward-backward system to a decoupled one and then by applying a standard comparison principle for backward SDEs. \qed

\vskip 6pt\noindent
\textbf{Impact on the model for emission control}. As expected, the previous result implies that the tougher the regulation (i.e. the larger $\lambda$ and/or the smaller $\Lambda$), the higher the emission reductions (the lower $\tE_t$). In particular, in the absence of regulation which corresponds to $\lambda=0$, the aggregate level of emissions is at its highest.

We also notice that the assumptions in Theorem \ref{thmsec3} can be specified in such a way that the aggregate perceived emission process $\tE$ takes non-negative values, as expected from the rationale of the model.  

\begin{proposition}
\label{prop:sign}
Let the conditions of Theorem \ref{thmsec3} hold true. Assume further that $f(0)=0$ and there exists $r>0$ such that $\sigma(t,0)=0$, $b(t,.) \ge 0$ on $[0,r]$, and $b(t,.)\le 0$ on $[-r,0]$. Then:
\\
{\rm (i)}\quad for any $\tilde E_0\ge 0$, the process $\tE$ in \eqref{fo:RNfbsde} is non-negative;
\\
{\rm (ii)}\quad if in addition $\tilde E_0>0$, then $\tE_t>0$ for all $t\in[0,T)$.
\end{proposition}

\begin{proof} 
By \eqref{f-Niz}, we know that $f(y) \geq 0$ for $y \in [0,\lambda]$. Since the process $(Y_{t})_{0 \leq t \leq T}$ is $[0,\lambda]$-valued, we deduce from the comparison principle for forward SDEs that the forward process $(\tilde E_{t})_{0 \leq t \leq T}$ is dominated by the solution $(X_{t})_{0 \leq t \leq T}$ to the SDE:
\begin{equation*}
X_{t} = \tilde{E}_{0} + \int_{0}^t b(s,X_{s}) ds 
              + \int_{0}^t \sigma (s,X_{s}) d\tilde{W}_{s}, \quad 0 \leq t \leq T. 
\end{equation*}
Observe that our conditions on $b$ and $\sigma$ imply that, whenever $\tilde E_{0} \le 0$, we have $X\le 0$ and therefore $\tilde E \leq 0$. Then 
$Y_{T} = \lambda {\mathbf 1}_{[\Lambda,+\infty)}(\tilde E_{T})=0$, so that $u(0,\tilde E_{0}) = {\mathbb E}(Y_{T})=0$. Similarly, $u(t,e)=0$, for any $t \in [0,T]$ and $e \leq 0$. 

As a consequence, for any initial condition $\tilde E_{0}$, we can write $(f(Y_{t}))_{0 \leq t < T}$ in the forward equation in \eqref{fo:RNfbsde} as 
\begin{equation*}
f(Y_{t}) = f(u(t,\tilde E_{t})) = f(u(t,\tilde E_{t})) - f(u(t,0)) = 
\frac{f(u(t,\tilde E_{t}))-f(u(t,0))}{\tilde E_{t}} \tilde E_{t} {\mathbf 1}_{\{\tilde E_{t} \not = 0\}}, 
\end{equation*}
where the ratio $(f(u(t,e))-f(u(t,0)))/e$, for $e \not =0$, is uniformly bounded in $e \in \RR \setminus \{0\}$ and in $t$ in compact subsets of $[0,T)$ since $u$ is Lipschitz-continuous in space, uniformly in time in compact subsets of $[0,T)$, see Point (ii-1) in the proof of Theorem \ref{thmsec3}. Similarly, the processes
 \begin{eqnarray*}
 \beta_t:=\frac{b(t,\tilde E_t)}{\tilde E_t}
          \mathbf{1}_{\{\tilde E_t\neq 0\}}
 &\mbox{and}&
 \Sigma_t
 :=
 \frac{\sigma(t,\tilde E_t)}{\tilde E_t}\mathbf{1}_{\{\tilde E_t\neq 0\}}
\end{eqnarray*}
are adapted and bounded, by the Lipschitz property of the coefficients $b,\sigma$ in $e$ uniformly in $t$, and the fact that $b(t,0)=\sigma(t,0)=0$. 
We then deduce that $(\tilde E_{t})_{0 \leq t < T}$ may be expressed as
\begin{equation*}
\tilde E_{t} 
= 
\tilde E_{0} \exp \biggl( \int_0^t(\beta_s-\varphi_s-\frac12 \Sigma^2_s)ds 
                          +\int_0^t \Sigma_s d\tilde W_s 
                  \biggr), \quad 0 \leq t < T,
\end{equation*}
with $\varphi_{t} = [f(Y_{t})/\tilde E_{t}] {\mathbf 1}_{\{\tilde E_{t} \not = 0\}}$, $0 \leq t < T$. 
\end{proof}

\begin{remark}
Using for $u$  additional estimates from the theory of partial differential equations, we may also prove that $\varphi_{t}$ appearing in the above proof of Proposition \ref{prop:sign} grows up at most as $(T-t)^{-1/2}$ when $t \nearrow T$. This implies that $\varphi$ is integrable on the whole $[0,T]$ and thus, that $\tilde E_{T} >0$ as well when $\tilde E_{0} >0$. Since this result is not needed in this paper, we do not provide a detailed argument.
\end{remark}

\begin{remark}
The non-degeneracy of $\sigma$ in the neighborhood of $(T,\Lambda)$, see \eqref{sigma-Niz}, is compatible with the condition $\sigma(t,0)=0$ of Proposition \ref{prop:sign}, since $\Lambda$, which is the regulatory emission cap in practice,  is expected to be (strictly) positive.
\end{remark}

\section{\textbf{Enlightening Example of a Singular FBSDE}}
\label{se:example}

We saw in the previous section that the terminal condition of the backward equation can 
be a discontinuous function of the terminal value of the forward component without threatening 
existence or uniqueness of a solution to the FBSDE when the forward dynamics are non-degenerate 
in the neighborhood of the singularity of the terminal condition. In this section, we show that this is not 
the case when the forward dynamics are degenerate, even if they are hypoelliptic and the solution of the forward equation has a density before maturity. 
We explained in the introduction why this seemingly pathological mathematical property should not come as a surprise in the context of 
equilibrium models for cap-and-trade schemes.
\vskip 4pt
Motivated by the second model given in subsection \ref{sub:elecprice},  we consider the FBSDE:
\begin{equation}
\label{eq:1}
\begin{cases}
&dP_t = dW_t,\\
&dE_t = \bigl(P_t  - Y_t\bigr) dt,\\
&dY_t = Z_t dW_t, \quad 0 \leq t \leq T,
\end{cases}
\end{equation}
with the terminal condition
\begin{equation}
\label{fo:terminal}
Y_T= {\mathbf 1}_{[\Lambda,\infty)}(E_T), 
\end{equation}
for some real number $\Lambda$. Here, $(W_t)_{t \in[0,T]}$ is a one-dimensional Wiener process.
This unrealistic model corresponds to quadratic costs of production,  and choosing appropriate units for the penalty $\lambda$ and the emission rate $\epsilon$ to be $1$. (For notational convenience, the martingale measure is denoted by $\PP$ instead of ${\mathbb Q}$ as in Section \ref{se:equilibrium}, and the associated Brownian motion by $(W_{t})_{0 \leq t  \leq T}$
instead of $(\tW_{t})_{0 \leq t \leq T}$).

Below, we won't discuss the sign of the emission process $E$ as we did in Proposition 
\ref{prop:sign} above for 
the first model.
Our interest in the example (\ref{eq:1})--(\ref{fo:terminal}) is the outcome of its mathematical analysis, not its realism!
We prove the following unexpected result.

\begin{theorem}
\label{thm:1}
Given $(p,e) \in \RR^2$, there exists a unique progressively measurable
triple $(P_t,E_t,Y_t)_{0 \leq t \leq T}$ satisfying 
\reff{eq:1} together with the initial conditions $P_0=p$ and $E_0=e$, and
\begin{equation}
\label{eq:0:1}
{\mathbf 1}_{(\Lambda,\infty)}(E_T) \leq Y_T \leq {\mathbf 1}_{[\Lambda,\infty)}(E_T).
\end{equation}
Moreover, the marginal distribution of $E_t$ is absolutely continuous with respect to the Lebesgue
measure for any $0 \leq t < T$, but has a Dirac mass at $\Lambda$ when $t=T$. In other words:
\begin{equation*}
\PP\{E_T=\Lambda\} >0.
\end{equation*}
\end{theorem}
In particular, $(P_t,E_t,Y_t)_{0 \leq t \leq T}$ may not
satisfy the terminal condition 
$\PP\{Y_T =  {\mathbf 1}_{[\Lambda,\infty)}(E_T)\}=1$.
 However, the weaker form \reff{eq:0:1} of  terminal condition
is sufficient to guarantee uniqueness. 

\vskip 4pt
Before we engage in the technicalities of the proof we notice that the transformation 
\begin{equation}
\label{eq:25:08:1}
(P_t,E_t)_{0 \leq t \leq T} \hookrightarrow (\bar{E}_t = E_t + (T-t) P_t)_{0 \leq t \leq T}
\end{equation}
maps the original FBSDE \reff{eq:1} into the simpler one
\begin{equation}
\label{eq:1:b}
\begin{cases}
&d \bar{E}_t = - Y_t dt+ (T-t) dW_t,\\
&d Y_t = Z_t dW_t,
\end{cases}
\end{equation}
with the same terminal condition $Y_T={\mathbf 1}_{[\Lambda,\infty)}(\bar{E}_T)$.
Moreover, the dynamics of $(E_t)_{0 \leq t \leq T}$ can be recovered from those of
$(\bar{E}_t)_{0 \leq t \leq T}$ since $(P_t)_{0 \leq t \leq T}$ in \reff{eq:1} is purely autonomous.
In particular, except for the proof of the absolute continuity of $E_t$ for $t  < T$, we restrict our analysis to the
proof of Theorem \ref{thm:1}, for $\bar E$ solution of \reff{eq:1:b} since $E$ and $\bar E$ have the same terminal values at 
time $T$.
\vskip 4pt

We emphasize that system \reff{eq:1:b} is doubly singular at maturity time $T$: the diffusion coefficient of the forward equation vanishes as $t$ tends to $T$ and the boundary condition of the backward equation is discontinuous at $\Lambda$. Together, both singularities make the emission process accumulate a non-zero mass at $\Lambda$ at time $T$. 
This phenomenon must be seen as a \emph{stochastic residual} of the shock wave observed 
in the inviscid Burgers equation 
\begin{equation}
\label{eq:21:08:1}
\partial_t v(t,e) - v(t,e) \partial_e v(t,e) = 0, \quad t \in [0,T), \ e \in \RR,
\end{equation}
with $v(T,e)={\mathbf 1}_{[\Lambda,+\infty)}(e)$ as boundary condition. 
As explained below, equation \reff{eq:21:08:1} is the first-order version 
of the second-order equation associated with \reff{eq:1:b}.
\vskip 4pt

Indeed, it is well-known that the characteristics of \reff{eq:21:08:1} may meet at time $T$ and at point $\Lambda$. By analogy, the trajectories of the forward process in \eqref{eq:1:b} may hit $\Lambda$ at time $T$ with a non-zero probability, then producing a Dirac mass. In other words, the shock phenomenon behaves like a trap
into which the process $({E}_t)_{0 \leq t \leq T}$ (or equivalently the process
$(\bar{E}_t)_{0 \leq t \leq T}$) may fall with a non-zero probability. 
It is then well-understood that the noise plugged into
the forward process $(\bar{E}_t)_{0 \leq t \leq T}$ may help it to escape the
trap. For example, we saw in Section \ref{se:equilibrium} that the emission process did not see the trap when it was strongly elliptic in the neighborhood of the singularity. In the current framework, the diffusion coefficient vanishes in a linear way as time tends to maturity: it decays too fast to prevent almost every realization of the process from falling into the trap. 
\vskip 4pt

As before, we prove existence of a solution to \reff{eq:1:b} by first smoothing
the singularity in the terminal condition, solving the problem for a smooth
terminal condition, and obtaining  a solution to the original problem by a limiting argument.
However, in order to prove the existence of a limit, we will use PDE a priori estimates and compactness arguments instead of comparison and monotonicity arguments.
We call \emph{mollified equation} the system \reff{eq:1:b} with a terminal condition
\begin{equation}
\label{eq:1:1}
Y_T = \phi(\bar{E}_T),
\end{equation}
given by a Lipschitz non-decreasing function $\phi$ from $\RR$ to $[0,1]$ which we view as an approximation of the indicator function
appearing in the terminal condition \reff{fo:terminal}.

\subsection{Lipschitz Regularity in Space}

\begin{proposition}
\label{prop:1:1}
Assume that the terminal condition in \reff{eq:1:b} is given by \reff{eq:1:1} with a Lipschitz non-decreasing function $\phi$ with values in $[0,1]$. Then, for each $(t_0,e)\in [0,T]\times\RR$,
\reff{eq:1:b} admits a unique solution $(\bar{E}_t^{t_0,e},Y_t^{t_0,e},Z_t^{t_0,e})_{t_0 \le t \le T}$ satisfying $\bar{E}_{t_0}^{t_0,e}=e$ and $Y_T^{t_0,e}=\phi(\bar{E}_T^{t_0,e})$ . Moreover, the mapping
\begin{equation*}
(t,e) \hookrightarrow v(t,e) = Y_t^{t,e}
\end{equation*}
is $[0,1]$-valued,
is of class ${\mathcal C}^{1,2}$ on $[0,T) \times \RR$
and has H\"older continuous first-order derivative in time and first and second-order
derivatives in space.

Finally, the H\"older norms of $v$, $\partial_e v$, $\partial_{e,e}^2 v$
and $\partial_t v$
on a given compact subset of $[0,T) \times \RR$ do not depend upon 
the smoothness of $\phi$ provided $\phi$ is $[0,1]$-valued and
non-decreasing. Specifically, the first-order derivative in space satisfies
\begin{equation}
\label{eq:23:08:1}
0 \leq 
\partial_ev(t,e) \leq \frac{1}{T-t}, \quad t \in [0,T).
\end{equation}
In particular, $e \hookrightarrow v(t,e)$ is non-decreasing for any $t \in [0,T)$.

Finally, for a given initial condition $(t_0,e)$, the processes $(Y_t^{t_0,e})_{t_0 \leq t \leq T}$ and 
$(Z_t^{t_0,e})_{t_0 \leq t < T}$, solution to the backward equation in \reff{eq:1:b} (with $\phi$ as boundary condition),
are given by:
\begin{equation}
\label{eq:1:1:2}
Y_t^{t_0,e} = v(t,\bar{E}_t^{t_0,e}), \ t_0 \leq t \leq T \ ; \quad 
Z_t^{t_0,e} = (T-t) \partial_e v(t,\bar{E}_t^{t_0,e}), \ t_0 \leq t < T.
\end{equation}
\end{proposition}

\begin{proof}
The problem is to solve the system 
\begin{equation}
\label{eq:2}
\begin{cases}
&d\bar{E}_t = -Y_t dt + (T-t) dW_t,\\
&dY_t = Z_t dW_t,
\end{cases}
\end{equation}
with $\xi=\phi(\bar E_T)$ as terminal condition and $(t_0,e)$ as initial condition.
The drift in the first equation, i.e. $(t,y) \in [0,T] \times \RR \hookrightarrow -y$, is decreasing in $y$, and Lipschitz continuous, uniformly in $t$. 
By Theorem 2.2 in Peng and Wu \cite{peng:wu} (with $G=1$, $\beta_1=0$ and $\beta_2=1$ therein), we know that equation \reff{eq:2} admits at most one solution.
Unfortunately, Theorem 2.6 in Peng and Wu (see also Remark 2.8 therein) does not apply to prove existence directly. 

To prove existence, we use a variation of the induction method in Delarue \cite{Delarue02}. In the whole argument, 
$t_{0}$ stands for the generic initial time at which the process $\bar{E}$ starts. The proof consists in extending the local solvability property of Lipschitz forward-backward SDEs as the distance $T-t_{0}$ increases, so that the value of $t_{0}$ will vary in the proof. Recall indeed from Theorem 1.1 in \cite{Delarue02} that existence and uniqueness hold in small time. Specifically, we can find some small positive real number $\delta$, possibly depending on the Lipschitz constant of $\phi$, such that \reff{eq:2} admits a unique solution when $t_{0}$ belongs to the interval $[T-\delta,T]$. Remember that the initial condition is $\bar{E}_{t_0}=e$. As a consequence, we can define the \emph{value function}
$v : [T-\delta,T] \times \RR\ni (t_0,e)  \hookrightarrow Y_{t_0}^{t_0,e}$. By Corollary 1.5 in \cite{Delarue02}, it is known to be Lipschitz in space uniformly in time as long as the initial time parameter $t_{0}$ remains in $[T-\delta,T]$. The diffusion coefficient $T-t$ in \eqref{eq:2} being uniformly bounded away from $0$ on the interval
$[0,T-\delta]$, by Theorem 2.6 in \cite{Delarue02},  \reff{eq:2}
admits a unique solution on $[t_0,T-\delta]$ when $t_{0}$ is assumed to be in $[0,T-\delta)$. Therefore, we can construct a solution to \reff{eq:2} in two steps when $t_{0} < T-\delta$: we first solve  \reff{eq:2} on $[t_0,T-\delta]$ with $\bar{E}_{t_0}=e$ as initial condition and $v(T-\delta,\cdot)$ as giving the terminal condition, the solution being denoted by $(\bar{E}_t,Y_t,Z_t)_{t_0 \leq t \leq T-\delta}$; then, we solve \reff{eq:2} on $[T-\delta,T]$ with the previous $\bar{E}_{T-\delta}$ as initial condition and with $\phi$ as giving the terminal condition, the solution being denoted by $(\bar{E}_t',Y_t',Z_t')_{T-\delta \leq t \leq T}$. We already know that 
$\bar{E}_{T-\delta}'$ matches $\bar{E}_{T-\delta}$. To patch $(\bar{E}_t,Y_t,Z_t)_{t_0 \leq t \leq T-\delta}$ and 
$(\bar{E}_t',Y_t',Z_t')_{T-\delta \leq t \leq T}$ into a single solution over the whole time interval $[t_0,T]$, it is sufficient to check the continuity property 
$Y_{T-\delta}=Y_{T-\delta}'$ as done in Delarue \cite{Delarue02}. This continuity property is a straightforward consequence of Corollary 1.5 in \cite{Delarue02}: 
on $[T-\delta,T]$, $(Y_t')_{T-\delta \leq t \leq T}$ has the form 
$Y_t'=v(t,\bar{E}_t')$. In particular, $Y_{T-\delta}'=v(T-\delta,\bar{E}_{T-\delta}')=v(T-\delta,\bar{E}_{T-\delta})=Y_{T-\delta}$. This proves the existence of a solution to \reff{eq:2} with $\bar{E}_{t_0}=e$ as initial condition.

We conclude that, for any $(t_0,e)$, \reff{eq:2} admits a unique solution $(\bar{E}_t^{t_0,e},Y_t^{t_0,e},Z_t^{t_0,e})_{t_0 \leq t \leq T}$ satisfying $\bar{E}_{t_0}^{t_0,e}=e$ and $Y_T^{t_0,e} = \phi(\bar{E}_T^{t_0,e})$. 
In particular, the value function $v : (t_0,e) \hookrightarrow Y_{t_0}^{t_0,e}$ (i.e. the value at time $t_0$ of the solution $(Y_t)_{t_0 \leq t \leq T}$ under the initial condition  $\bar{E}_{t_0}=e$) can be defined on the whole $[0,T] \times \RR$.

From Corollary 1.5 in \cite{Delarue02} and the discussion above, we know that
the mapping $e \hookrightarrow v(t,e)$ is Lipschitz continuous 
when $T-t$ is less than $\delta$ and that, for any $t_0 \in [0,T]$,
$Y_t^{t_0,e}$ has the form $Y_t^{t_0,e} = v(t,\bar{E}_t^{t_0,e})$
when $T-t$ is less than $\delta$. In particular, on any $[0,T-\delta']$, $\delta'$ being less than $\delta$,  \reff{eq:2} may be seen as a uniformly elliptic FBSDE with a Lipschitz boundary condition.
By Theorem 2.1 in Delarue and Guatteri \cite{DelarueGuatteri06} (together with the discussion in Section 8 therein), we deduce that $v$ belongs to  ${\mathcal C}^{0}([0,T] \times \RR) \cap
{\mathcal C}^{1,2}([0,T) \times \RR)$, that $t \hookrightarrow 
\|\partial_{e} v(t,\cdot)\|_{\infty}$ is bounded on the whole $[0,T]$ and that $t \hookrightarrow \|\partial_{ee}^2 v(t,\cdot)\|_{\infty}$
is bounded on every compact subset of $[0,T)$\footnote{Specifically, Theorem 2.1 in \cite{DelarueGuatteri06} says that $v$ belongs to ${\mathcal C}^0([0,T) \times \RR)$ and that
$t \hookrightarrow \|\partial_{e} v(t,\cdot)\|_{\infty}$ is bounded on every compact subset of $[0,T)$.
In fact, by Corollary 1.5 in Delarue \cite{Delarue02}, we know that
$v$ belongs to ${\mathcal C}^0([T-\delta,T] \times \RR)$ and that
$t \hookrightarrow \|\partial_{e} v(t,\cdot)\|_{\infty}$ is bounded on $[T-\delta,T]$ for
$\delta$ small enough.}. Moreover, 
\reff{eq:1:1:2} holds.

By the martingale property of $(Y_t^{t_0,e})_{t_0 \leq t \leq T}$, it is well-seen that
$v$ is $[0,1]$-valued. To prove that it is non-decreasing (with respect to $e$), we follow the proof of Theorem \ref{thmsec3}. We notice that $(\bar{E}_t^{t_0,e})_{t_0 \leq t \leq T}$ satisfies the SDE:
\begin{equation*}
d \bar{E}_t^{t_0,e} = - v(t,\bar{E}_t^{t_0,e}) dt + (T-t) dW_t, \quad t_0 \leq t \leq T,
\end{equation*}
which has a Lipschitz drift with respect to the space variable. In particular, for $e \leq e'$, $\bar{E}_T^{t_0,e} \leq \bar{E}_T^{t_0,e'}$, so that 
$v(t_0,e) = {\mathbb E}\phi(\bar{E}_T^{t_0,e}) \leq 
{\mathbb E}\phi(\bar{E}_T^{t_0,e'}) = v(t_0,e')$.

\vskip 2pt
We now establish \reff{eq:23:08:1}.
For $t_0 \leq t \leq T$, the forward equation in \reff{eq:2} has the form
\begin{equation}
\label{eq:25:08:2}
\bar{E}_t^{t_0,e} = e - \int_{t_0}^t v(s,\bar{E}_s^{t_0,e}) ds + \int_{t_0}^t 
(T-s) dW_s.
\end{equation}
Since $v$ is ${\mathcal C}^1$ in space on $[0,T) \times \RR$ with bounded Lipschitz first-order derivative, we
can apply standard results on the differentiability of stochastic flows 
(see for example Kunita's monograph \cite{Kunita.book}). 
We deduce that, for almost every realization of the randomness and
for any $t \in [t_0,T)$, the mapping $e \hookrightarrow \bar{E}_t^{t_0,e}$ is differentiable and
\begin{equation}
\label{fo:partialE}
\partial_e \bar{E}_t^{t_0,e} = 1 - \int_{t_0}^t \partial_e v(s,\bar{E}_s^{t_0,e})
\partial_e \bar{E}_s^{t_0,e} ds.
\end{equation}
In particular,
\begin{equation}
\label{eq:19:04:100}
\partial_e \bar{E}_t^{t_0,e} = \exp \biggl( - \int_{t_0}^t
\partial_e v(s,\bar{E}_s^{t_0,e}) ds \biggr).
\end{equation}
Since $v$ is non-decreasing, we know that $\partial_e v \geq 0$
on $[0,T) \times \RR$
so that $\partial_e \bar{E}_t^{t_0,e}$ belongs to $[0,1]$.
Since $\partial_e v$ is also bounded on the whole $[0,T) \times \RR$, we deduce
by differentiating the right-hand side in \reff{eq:25:08:2} with $t=T$ that $\partial_e \bar{E}_T^{t_0,e}$ exists as well and that $\partial_e \bar{E}_T^{t_0,e}
= \lim_{t \rightarrow T} \partial_e \bar{E}_t^{t_0,e} \in [0,1]$.
To complete the proof of \reff{eq:23:08:1}, we then notice that for any $t \in [t_0,T]$,
\begin{equation*}
d \bigl[ (T-t) Y_t^{t_0,e} - \bar{E}_t^{t_0,e} \bigr]
= (T-t) dY_t^{t_0,e} - (T-t) dW_t = (T-t) [Z_t^{t_0,e} - 1] dW_t,
\end{equation*}
so that taking expectations we get:
 \begin{equation*}
 (T-t_0) v(t_0,e) - e = - {\mathbb E} \bigl[\bar{E}_T^{t_0,e} \bigr].
 \end{equation*}
Now, differentiating with respect to $e$, we have:
 \begin{equation*}
 (T-t_0) \partial_e v(t_0,e) = 1 - {\mathbb E}
  \bigl[\partial_e \bar{E}_T^{t_0,e} \bigr] \leq 1,
  \end{equation*}
which concludes the proof of \reff{eq:23:08:1}. 

It now remains to investigate the H\"older norms (both in time and space) of $v$, $\partial_e v$, $\partial_{ee}^2 v$ and $\partial_t v$.
\noindent
We first deal with $v$ itself. For $0 < t < s < T$,
\begin{equation*}
\begin{split}
v(s,e) - v(t,e) &= v(s,e) - v(s,\bar{E}_s^{t,e}) + v(s,\bar{E}_s^{t,e})
- v(t,e)
\\
&= v(s,e) - v(s,\bar{E}_s^{t,e}) + {Y}_s^{t,e} - {Y}_t^{t,e}
\\
&= v(s,e) - v(s,\bar{E}_s^{t,e}) + \int_t^s {Z}_r^{t,e} dB_r.
\end{split}
\end{equation*}
From \reff{eq:23:08:1}, we deduce
\begin{equation*}
\begin{split}
|v(s,e) - v(t,e)| &\leq \frac{1}{T-s} {\mathbb E}\bigl| \bar{E}_s^{t,e} - e \bigr|
+ {\mathbb E} \biggl| \int_t^s {Z}_r^{t,e} dB_r \biggr|
\\
&\leq \frac{1}{T-s} \biggl[ s-t + \biggl(\int_t^s (T-r)^2 dr \biggr)^{1/2} \biggr]
 + {\mathbb E} \biggl[ \int_t^s |{Z}_r^{t,e}|^2 dr \biggr]^{1/2}
\\
&\leq \frac{1}{T-s} \biggl[ s-t + \biggl(\int_t^s (T-r)^2 dr \biggr)^{1/2} \biggr]
+ (s-t)^{1/2},
\end{split}
\end{equation*}
since ${Z}_r^{t,e} = (T-r) \partial_e v(r,\bar{E}_r^{t,e}) \in [0,1]$.
So for $\epsilon>0$, $v$ is 1/2-H\"older continuous in time
$t \in [0,T-\epsilon]$, uniformly in space and in the smoothness of $\phi$.

Now, by Theorem 2.1 in Delarue and Guatteri \cite{DelarueGuatteri06}, we know that
$v$ satisfies the PDE
\begin{equation}
\label{eq:3}
\begin{split}
&\partial_t v (t,e) + \frac{(T-t)^2}{2} \partial^2_{ee}v(t,e) - v(t,e)  \partial_ev(t,e)=0, \quad t \in [0,T),  \ e \in \RR,
\end{split}
\end{equation}
with $\phi$ as boundary condition.  
On $[0,T-\epsilon] \times \RR$, $\epsilon >0$, equation \reff{eq:3} is a non-degenerate second-order PDE of dimension 1 with 
 $- v$ as drift, this drift being ${\mathcal C}^{1/2,1}$-continuous independently of the smoothness of $\phi$. By well-known results in 
 PDEs (so called Schauder estimates, see for example Theorem 8.11.1 in Krylov \cite{Krylov}), for any small $\eta>0$, the ${\mathcal C}^{(3-\eta)/2,3-\eta}$-norm of $v$ on $[0,T-\epsilon] \times \RR$ is independent of the smoothness of  $\phi$.  \end{proof}

\begin{remark}
\label{rem:24:08:1}
As announced, equation \reff{eq:3} is of Burgers type. In particular, it has the same first-order part as equation \reff{eq:21:08:1}.
\end{remark}

\subsection{Boundary Behavior}
Still in the framework of a terminal condition given by a smooth (i.e. 
non-decreasing Lipschitz) function with values in $[0,1]$, we investigate the shape of the solution as $t$ approaches $T$.

\begin{proposition} 
\label{prop:1:3}
Assume that there exists some real $\Lambda^+$ such that $\phi(e)=1$ on $[\Lambda^+,+\infty)$. Then, there exists a universal constant $c>0$ such that for any $\delta >0$
\begin{equation}
\label{eq:22:3:1}
v\bigl(t,\Lambda^++T-t+\delta\bigr) \geq 1- 
  \exp \bigl( - c \frac{\delta^2}{(T-t)^3} \bigr), \quad 0 \leq t < T.
\end{equation}
In particular, $v(t,e) \rightarrow 1$ as $t
\nearrow T$ uniformly in $e$ in compact subsets of 
$(\Lambda^+,+\infty)$.

Similarly, assume that there exists an interval $(-\infty,\Lambda^-]$ such that $\phi(e)=0$ on $(-\infty,\Lambda^-]$. Then, for any $\delta >0$,
\begin{equation}
\label{eq:21}
v(t,\Lambda^--\delta) \leq 
  \exp \bigl( - c \frac{\delta^2}{(T-t)^3} \bigr).
\end{equation}
In particular, $v(t,e) \rightarrow 0$ as $t \nearrow T$ uniformly in $e$ in compact subsets of 
$(-\infty,\Lambda^-)$.
\end{proposition}

\begin{proof} We only prove \reff{eq:22:3:1}, the proof of \reff{eq:21} being similar. To do so, we fix 
$(t_0,e) \in [0,T) \times \RR$ and consider the following system
\begin{equation*}
\begin{cases}
&dE_t^- = - dt + (T-t) dW_t\\
&dY_t^{-} = Z_t^{-} dW_t, \quad t_0 \leq t \leq T,
\end{cases}
\end{equation*}
with $E_{t_0}^- = e$ as initial condition for the forward equation
and $Y_T^- = \phi(E_T^-)$ as terminal condition for the backward part. 
The solution $(\bar{E}_t^{t_0,e},Y_t^{t_0,e},Z_t^{t_0,e})_{t_0 \leq t \leq T}$
given by Proposition \ref{prop:1:1} with $\bar{E}_{t_0}^{t_0,e}=e$ and $Y_T^{t_0,e} = \phi(\bar{E}_T^{t_0,e})$ satisfies $Y_t^{t_0,e} \in [0,1]$ for any $t \in [t_0,T]$ so that 
$E_t^- \leq \bar{E}_t^{t_0,e}$ almost surely for $t \in [t_0,T]$.
Now, since $\phi$ is non-decreasing, 
$\phi(E_T^-) \leq \phi(\bar{E}_T^{t_0,e})$ almost surely, namely
$Y_{t_0}^- \leq Y_{t_0}^{t_0,e}$. Setting $v^-(t_0,e) =Y_{t_0}^{-}$, recall that $Y_{t_0}^{-}$ is deterministic, we see that:
\begin{equation}
\label{eq:1:1:1}
v^-(t_0,e) \leq v(t_0,e) \leq 1.
\end{equation}
Now, since
\begin{equation*}
v^-(t_0,e) = {\mathbb E} \phi(E_T^-) =  {\mathbb E} \phi\biggl( e- (T-t_0) + \int_{t_0}^T (T-s) dW_s \biggr)
\end{equation*}
with $\phi \geq {\mathbf 1}_{[\Lambda^+,+\infty)}$,
by choosing $e= \Lambda^++(T-t_0) + \delta$ as in the statement of Proposition 
\ref{prop:1:3} we get:
\begin{equation*}
\begin{split}
{\mathbb E} \phi(E_T^-) &=
{\mathbb E}  \phi \biggl( \Lambda^+ + \delta + \int_{t_0}^T (T-s) dW_s\biggr) 
\\
&\geq {\mathbb P} \biggl[ \Lambda^+ + \delta + \int_{t_0}^T (T-s) dW_s\geq \Lambda^+ \biggr]
\\
&=  \PP \biggl[ \int_{t_0}^T (T-s) dW_s  \geq -\delta \biggr]= 1 - \PP \biggl[ \int_{t_0}^T (T-s) dW_s  \leq -\delta \biggr]
\end{split}
\end{equation*}
and we complete the proof by applying standard estimates for the decay of the cumulative distribution function 
of a Gaussian random variable. Note indeed that
$\text{var}(\int_{t_0}^T (T-s) dW_s)=(T-t_0)^3/3$ if we use the notation $\text{var}(\xi)$ for the variance of 
a random variable $\xi$.
\end{proof}

The following corollary elucidates the boundary behavior between $\Lambda^-$ and
$\Lambda^+ + (T-t)$ with $\Lambda^-$ and $\Lambda^+$ as above.

\begin{corollary}
\label{corol:1:3}
Choose $\phi$ as in Proposition \ref{prop:1:3}. If 
there exists an interval $[\Lambda^+,+ \infty)$ on which $\phi(e)=1$, then
for $\alpha>0$ and $e < \Lambda^+ + (T-t) + (T-t)^{1+\alpha}$ we have:
\begin{equation}
\label{eq:01:04:1}
v(t,e) \geq \frac{e-\Lambda^+}{T-t} -  \exp \bigl( - \frac{c}{(T-t)^{1-2\alpha}} \bigr) - 
 (T-t)^{\alpha},
\end{equation}
for the same $c$ as in the statement of Proposition \ref{prop:1:3}.

Similarly, 
if there exists an interval $(-\infty,\Lambda^-]$ on which $\phi(e)=0$, then
for $\alpha>0$ and $e > \Lambda^-  - (T-t)^{1+\alpha}$ we have:
\begin{equation}
\label{eq:01:04:2}
v(t,e) \leq  \frac{e-\Lambda^-}{T-t} +  \exp \bigl( - \frac{c}{(T-t)^{1-2\alpha}} \bigr) + 
 (T-t)^{\alpha}.
\end{equation}
\end{corollary}

\begin{proof} We first prove \reff{eq:01:04:1}. 
Since $v(t,\cdot)$ is $1/(T-t)$ Lipschitz continuous, we have:
\begin{equation*}
\begin{split}
v\bigl(t,\Lambda^++(T-t)+(T-t)^{1+\alpha} \bigr) - v(t,e) &\leq \frac{\Lambda^+ -e + (T-t)+ (T-t)^{1+\alpha}}{T-t} 
\\
&=  \frac{\Lambda^+  - e}{T-t} + 1 + (T-t)^{\alpha}.
\end{split}
\end{equation*}
Therefore,
\begin{equation*}
v(t,e) \geq v\bigl(t,\Lambda^++(T-t)+(T-t)^{1+\alpha} \bigr) - 1 - (T-t)^{\alpha}
- \frac{\Lambda^+ - e}{T-t},
\end{equation*}
and applying \reff{eq:22:3:1} 
\begin{equation*}
v(t,e) \geq \frac{e-\Lambda^+}{T-t} - \exp \bigl(- c(T-t)^{2\alpha-1} \bigr) - (T-t)^{\alpha}.
\end{equation*}
For the upper bound, we use the same strategy. We start from
\begin{equation*}
v(t,e) - v\bigl(t,\Lambda^--(T-t)^{1+\alpha} \bigr) \leq
\frac{e-\Lambda^-}{T-t} + (T-t)^{\alpha},
\end{equation*}
so that
\begin{equation*}
v(t,e) \leq \frac{e-\Lambda^-}{T-t} + \exp \bigl(- c(T-t)^{2\alpha-1} \bigr) + (T-t)^{\alpha}.
\end{equation*}
\end{proof}

\noindent

\subsection{Existence of a Solution}
We now establish the existence of a solution 
to \reff{eq:1:b} with the original terminal condition.
We use a compactness argument giving the existence of a \emph{value function} for the problem.

\begin{proposition}
\label{prop:1:6}
There exists a continuous function $v : [0,T) \times \RR \hookrightarrow [0,1]$ satisfying
\begin{enumerate}
\item $v$ belongs to ${\mathcal C}^{1,2}([0,T) \times \RR)$ and solves \reff{eq:3}, 
\item $v(t,\cdot)$ is non-decreasing and $1/(T-t)$-Lipschitz continuous 
for any $t \in [0,T)$,
\item $v$ satisfies \reff{eq:22:3:1} and \reff{eq:21} with $\Lambda^-=\Lambda^+=\Lambda$,
\item $v$ satisfies \reff{eq:01:04:1} and \reff{eq:01:04:2} with $\Lambda^-=\Lambda^+=\Lambda$,
\end{enumerate}
and for any initial condition $(t_0,e) \in [0,T) \times \RR$,
the strong solution $(\bar{E}_t^{t_0,e})_{t_0 \leq t < T}$ of
\begin{equation}
\label{eq:05:04:1}
\bar{E}_t = e - \int_{t_0}^t v(s,\bar{E}_s) ds + \int_{t_0}^t (T-s) dW_s, \quad
t_0 \leq t < T,
\end{equation}
is such that $(v(t,\bar{E}_t^{t_0,e}))_{t_0 \leq t < T}$ is a martingale
with respect to the filtration generated by $W$.
\end{proposition}

\begin{proof} 
Choose a sequence of $[0,1]$-valued smooth non-decreasing functions $(\phi^n)_{n \geq 1}$ 
such that $\phi^n(e)=0$ for $e\leq \Lambda-1/n$ and $\phi^n(e)=1$
for $e \geq \Lambda+1/n$, $n \geq 1$, and denote by $(v^n)_{n \geq 1}$
the corresponding sequence of functions given by Proposition \ref{prop:1:1}.
By Proposition \ref{prop:1:1}, we can extract a subsequence, which we will still index
by $n$, converging uniformly on compact subsets
of $[0,T) \times \RR$. We denote by $v$ such a limit. Clearly, $v$ satisfies (1) in the statement
of Proposition \ref{prop:1:6}.
Moreover, it also satisfies (2) because of Proposition \ref{prop:1:1}, (3) by Proposition \ref{prop:1:3}, 
and (4) by Corollary \ref{corol:1:3}. Having Lipschitz coefficients, the stochastic differential equation \reff{eq:05:04:1} has a unique strong solution on $[t_0,T)$ for any initial condition $\bar E_{t_0}=e$. If we denote the solution by $(\bar{E}_t^{t_0,e})_{t_0 \leq t <T}$, It\^o's formula and \reff{eq:3}, imply that the process
$(v(t,\bar{E}_t^{t_0,e}))_{t_0 \leq t < T}$ is a local martingale. Since it is bounded, it is a bona fide martingale.
\end{proof}

\noindent
We finally obtain the desired solution to the FBSDE in the sense of Theorem \ref{thm:1}.

\begin{proposition}
\label{prop:1:7}
$v$ and $(\bar{E}_t^{t_0,e})_{t_0 \leq t <T}$ being as above and setting
\begin{equation*}
Y_t^{t_0,e} = v(t,\bar{E}_t^{t_0,e}), \ 
Z_t^{t_0,e} = (T-t) \partial_e v(t,\bar{E}_t^{t_0,e}), \quad t_0 \leq t < T,
\end{equation*}
the process $(\bar{E}_t^{t_0,e})_{t_0 \leq t < T}$ has an a.s. limit
$\bar{E}_T^{t_0,e}$ as $t$ tends to $T$. Similarly, the process $(Y_t^{t_0,e})_{t_0 \leq t < T}$ has an a.s. limit
$Y_T^{t_0,e}$ as $t$ tends to $T$, and the extended process $(Y_t^{t_0,e})_{t_0 \leq t \leq T}$
is a martingale with respect to the filtration generated by $W$. Morever, $\PP$-a.s., we have:
\begin{equation}
\label{eq:05:04:2}
{\mathbf 1}_{(\Lambda,\infty)}(\bar{E}_T^{t_0,e}) \leq Y_T^{t_0,e} \leq {\mathbf 1}_{[\Lambda,\infty)}(\bar{E}_T^{t_0,e}).
\end{equation}
and
\begin{equation}
\label{eq:05:04:3}
Y_T^{t_0,e} = Y_{t_0}^{t_0,e} + \int_{t_0}^T Z_t^{t_0,e} dW_t,
\end{equation}
\end{proposition}
\noindent
Notice that $Z_t^{t_0,e}$ is not defined for $t=T$.

\begin{proof}
The proof is straightforward now that we have collected all the necessary ingredients.
We start with the extension of $(\bar{E}_t^{t_0,e})_{t_0 \leq t < T}$ up to time $T$.
The only problem is to extend the drift part in \reff{eq:05:04:1}, but since $v$ is non-negative and bounded, it is clear that the process
$$
\biggl(\int_{t_0}^t v(s,\bar{E}_s^{t_0,e})ds \biggr)_{t_0 \leq t < T}
$$
is almost-surely increasing in $t$, so that the limit exists.
The extension of $(Y_t^{t_0,e})_{t_0 \leq t < T}$ up to time $T$ follows
from the almost-sure convergence theorem for positive martingales. 

To prove \reff{eq:05:04:2}, we apply (3) in the statement of Proposition 
\ref{prop:1:6}. If $\bar{E}_T^{t_0,e} = \lim_{t \rightarrow T} \bar{E}_t^{t_0,e} >\Lambda$, then
we can find some $\delta >0$ such that
$\bar{E}_t^{t_0,e} > \Lambda + (T-t) + \delta$ for $t$ close to $T$, so that
$Y_t^{t_0,e} = v(t,\bar{E}_t^{t_0,e}) \geq 1 - \exp[-c \delta^2/(T-t)^3]$
for $t$ close to $T$, i.e. $Y_T^{t_0,e} \geq 1$. Since $Y_T^{t_0,e} \leq 1$, we deduce that
\begin{equation*}
\bar{E}_T^{t_0,e} > \Lambda \Rightarrow Y_T^{t_0,e} = 1.
\end{equation*}
In the same way,
\begin{equation*}
\bar{E}_T^{t_0,e} < \Lambda \Rightarrow Y_T^{t_0,e} = 0.
\end{equation*}
This proves 
\reff{eq:05:04:2}. Finally \reff{eq:05:04:3} follows from It\^o's formula. Indeed,
by It\^o's formula and \reff{eq:3},
\begin{equation*}
Y_t^{t_0,e} = Y_{t_0}^{t_0,e} + \int_{t_0}^t Z_s^{t_0,e} dW_s, \quad t_0 \leq t < T.
\end{equation*}
By definition, $Z_s^{t_0,e} = (T-s) \partial_e v(s,\bar{E}_s^{t_0,e})$, $t_0 \leq s < T$.
By part (2) in the statement of Proposition \ref{prop:1:6}, it is in $[0,1]$. Therefore, the It\^o integral
\begin{equation*}
\int_{t_0}^T Z_s^{t_0,e} dW_s
\end{equation*}
makes sense as an element of $L^2(\Omega,\PP)$. This proves
\reff{eq:05:04:3}. 
\end{proof}

\subsection{Improved Gradient Estimates}
Using again standard results on the differentiability of stochastic flows 
(see again Kunita's monograph \cite{Kunita.book}) we see that formulae \reff{fo:partialE} and \reff{eq:19:04:100}
still hold in the present situation of a discontinuous terminal condition.
We also prove a representation for the gradient of $v$ of Malliavin-Bismut type.
\begin{proposition}
\label{prop:26:08:1}
For $t_0 \in [0,T)$, $\partial_e v(t_0,e)$ admits the representation
\begin{equation}
\label{eq:19:4:7}
\partial_e v(t_0,e)
= 2 (T-t_0)^{-2} {\mathbb E} \biggl[ \lim_{\delta \rightarrow 0} v\bigl(T-\delta,\bar{E}_{T-\delta}^{t_0,e}\bigr)
\int_{t_0}^{T} 
\partial_e \bar{E}_t^{t_0,e} dW_t \biggr].
\end{equation}
In particular, there exists some constant $A>0$ such that
\begin{equation}
\label{eq:19:4:6}
\sup_{|e| >A} \sup_{0 \leq t \leq T} \partial_e v(t,e) < + \infty.
\end{equation}
\end{proposition}

\begin{proof} 
For $\delta>0$, Proposition \ref{prop:1:7} yields
\begin{equation*}
\begin{split}
{\mathbb E} \biggl[  v\bigl(T-\delta,\bar{E}_{T-\delta}^{t_0,e}\bigr)
\int_{t_0}^{T} 
\partial_e \bar{E}_t^{t_0,e} dW_t \biggr]
&= {\mathbb E} \biggl[  \int_{t_0}^{T-\delta}  Z_{t}^{t_0,e} dW_t
\int_{t_0}^{T} 
\partial_e \bar{E}_t^{t_0,e} dW_t \biggr]
\\
&= {\mathbb E} \biggl[  \int_{t_0}^{T-\delta}  (T-t) \partial_e v\bigl(t,\bar{E}_{t}^{t_0,e}\bigr) \partial_e \bar{E}_t^{t_0,e} dt\biggr].
\end{split}
\end{equation*}
The bounds we have on $\partial_e v$ and 
$(\partial_e \bar{E}_t^{t_0,e})_{t_0 \leq t < T}$ justify the exchange of the expectation and integral signs. We obtain:
\begin{equation*}
{\mathbb E} \biggl[  v\bigl(T-\delta,\bar{E}_{T-\delta}^{t_0,e}\bigr)
\int_{t_0}^{T} 
\partial_e \bar{E}_t^{t_0,e} dW_t \biggr]
= \int_{t_0}^{T-\delta} (T-t) {\mathbb E} \bigl[ \partial_e
\bigl[ v\bigl(t,\bar{E}_{t}^{t_0,e}\bigr) \bigr] \bigr] dt.
\end{equation*}
Similarly, we can exchange the expectation and the partial derivative so that
\begin{equation*}
{\mathbb E} \biggl[  v\bigl(T-\delta,\bar{E}_{T-\delta}^{t_0,e}\bigr)
\int_{t_0}^{T} 
\partial_e \bar{E}_t^{t_0,e} dW_t \biggr]
= \int_{t_0}^{T-\delta} (T-t) \partial_e \bigl[ {\mathbb E}  v\bigl(t,\bar{E}_{t}^{t_0,e}\bigr) \bigr]  dt.
\end{equation*}
Since $(v(t,\bar{E}_t^{t_0,e}))_{t_0 \leq t \leq T-\delta}$ is a martingale, 
we deduce:
\begin{eqnarray*}
{\mathbb E} \biggl[  v\bigl(T-\delta,\bar{E}_{T-\delta}^{t_0,e}\bigr)
\int_{t_0}^{T} 
\partial_e \bar{E}_t^{t_0,e} dW_t \biggr]
&=& \partial_e v(t_0,e) \int_{t_0}^{T-\delta} (T-t)  dt\\
&=& \frac{1}{2} (T-\delta-t_0)(T+\delta-t_0) \partial_e v(t_0,e).
\end{eqnarray*}
Letting $\delta$ tend to zero and applying dominated convergence, we 
complete the proof of the representation formula of the gradient.

To derive the bound \reff{eq:19:4:6}, we emphasize that, 
for $e$ away
from $\Lambda$ (say for example $e\ll \Lambda$), the probability that 
$(\bar{E}_t^{t_0,e})_{t_0 \leq t \leq T}$ hits $\Lambda$ is very small and decays
exponentially fast as $T-t_0$ tends to $0$. On the complement, i.e. for
$\sup_{t_0 \leq t \leq T} \bar{E}_t^{t_0,e} < \Lambda$, we know that
$v(t,\bar{E}_t^{t_0,e})$ tends to $0$ as $t$ tends to $T$. 
 Specifically,
following the proof of Proposition \ref{prop:1:3}, there exists 
a universal constant $c'>0$ such that
for any $e \leq \Lambda-1$ and $t_0 \in [0,T)$
\begin{equation*}
\begin{split}
(T-t_0)^2 \partial_e v(t_0,e) &\leq 2 (T-t_0)^{1/2} \PP^{1/2} \bigl[
\sup_{t_0 \leq t \leq T} \bar{E}_t^{t_0,e} \geq \Lambda \bigr]
\\
&\leq 2 (T-t_0)^{1/2} \PP^{1/2} \bigl[ \Lambda - 1  + \sup_{t_0 \leq t \leq T}
\int_{t_0}^t (T-s) dW_s \geq \Lambda \bigr]
\\
&\leq 2 (T-t_0)^{1/2} \PP^{1/2} \bigl[  \sup_{t_0 \leq t \leq T}
\int_{t_0}^t (T-s) dW_s \geq 1 \bigr]
\\
&\leq 2 (T-t_0)^{1/2} \exp \bigl( - \frac{c'}{(T-t_0)^3} \bigr),
\end{split}
\end{equation*}
the last line following from maximal inequality (IV.37.12) in Rogers and Williams
\cite{RogersWilliams}.

The same argument holds for $e>\Lambda+2$ by noting that  \reff{eq:19:4:7}
also holds for $v-1$. 
\end{proof}

\begin{remark}
\label{rem:24:08:2}
The stochastic integral in the
Malliavin-Bismut formula \reff{eq:19:4:7} is at most of order $(T-t_0)^{1/2}$.
Therefore, the typical resulting bound for $\partial_e v(t,e)$ in the neighborhood of $(T,\Lambda)$ is 
$(T-t)^{-3/2}$. Obviously, it is less accurate than the bound given by Propositions \ref{prop:1:1} and \ref{prop:1:6}.
This says that the Lipschitz smoothing of the singularity of the boundary condition
obtained in Propositions \ref{prop:1:1} and \ref{prop:1:6}, namely $\partial_e v(t,e) \leq (T-t)^{-1}$, follows from the first-order Burgers structure of the PDE \reff{eq:3} and that the diffusion term plays no role in it. 
This is a clue to understand why the diffusion process $\bar{E}$ feels the trap made by the boundary condition. On the opposite, the typical bound for $\partial_e v(t,e)$ we would obtain in the uniformly elliptic case by applying a Malliavin-Bismut formula (see Exercice 2.3.5 in Nualart \cite{Nualart}) is of order $(T-t)^{-1/2}$, which is much better
than $(T-t)^{-1}$.

Nevertheless, the following proposition shows that
the diffusion term permits to improve the bound obtained in 
Propositions \ref{prop:1:1} and \ref{prop:1:6}. Because of the noise plugged into 
$\bar{E}$, the bound $(T-t)^{-1}$ cannot be achieved. 
This makes a real difference with the inviscid Burgers equation
\reff{eq:21:08:1} which admits 
$$(t,e) \in [0,T) \times \RR 
\hookrightarrow \psi\bigl( \frac{e-\Lambda}{T-t} \bigr),$$
as solution, with $\psi(e)=1\wedge e^+$ for $e\in\RR$. (See 
for example (10.12') in 
Lax \cite{Lax}.)
\end{remark}

We thus prove the following stronger version of  Propositions \ref{prop:1:1} and \ref{prop:1:6}:

\begin{proposition}
\label{prop:1:7:b}
For any $(t_0,e) \in [0,T) \times \RR$, it holds $(T-t_0) \partial_e v(t_0,e) < 1$. 
\end{proposition}

\begin{proof}
Given $(t_0,e) \in [0,T) \times \RR$, we consider $(\bar{E}_t^{t_0,e},Y_t^{t_0,e},Z_t^{t_0,e})_{t_0 \leq t \leq T}$ as in the statement of Proposition \ref{prop:1:7}.
As in the proof of Proposition \ref{prop:1:1}, we start from
\begin{equation*}
d \bigl[ (T-t) Y_t^{t_0,e} - \bar{E}_t^{t_0,e} \bigr]
= (T-t) dY_t^{t_0,e} - (T-t) dW_t = (T-t) [Z_t^{t_0,e} - 1] dW_t, \quad t_0 \leq t  <T.
\end{equation*}
Therefore,  for any initial condition $(t_0,e)$,
\begin{equation*}
 (T-t_0) v(t_0,e) - e = - {\mathbb E} \bigl[\bar{E}_T^{t_0,e} \bigr].
\end{equation*}
Unfortunately, we do not know whether $\bar{E}^{t_0,e}_T$
 is differentiable with respect to $e$. However,
\begin{equation*}
\begin{split}
 (T-t_0) \partial_e v(t_0,e) &= 1 - 
 \lim_{h \rightarrow 0} h^{-1} {\mathbb E}
  \bigl[\bar{E}_T^{t_0,e+h} - \bar{E}_T^{t_0,e} \bigr]
  \\
  &= 1 - 
 \lim_{h \rightarrow 0} h^{-1} \lim_{t \nearrow T} {\mathbb E}
  \bigl[\bar{E}_t^{t_0,e+h} - \bar{E}_t^{t_0,e} \bigr]
  \leq 1 - \lim_{h \rightarrow 0}  \lim_{t \nearrow T}
  \inf_{|u| \leq h} {\mathbb E} \bigl[ \partial_e \bar{E}_t^{t_0,e+ u} \bigr]
\end{split}
\end{equation*}
Using \reff{eq:19:04:100}, the non-negativity of $\partial_e v$ and
Fatou's lemma,
\begin{equation*}
\begin{split}
 (T-t_0) \partial_e v(t_0,e) &\leq 1 -  \lim_{h \rightarrow 0}  
 \lim_{t \nearrow T}
   \inf_{|u| \leq h} {\mathbb E} \biggl[ \exp \biggl( - \int_{t_0}^t \partial_e v(s,\bar{E}_s^{t_0,e+u}) ds \biggr) \biggr]
  \\
  &\leq 1 -  \lim_{h \rightarrow 0}  
   \inf_{|u| \leq h} {\mathbb E}  \biggl[ \exp \biggl( - \int_{t_0}^T \partial_e v(s,\bar{E}_s^{t_0,e+u}) ds \biggr) \biggr]
   \\
   &\leq 1 -   {\mathbb E}  \biggl[ \exp \biggl( - 
   \lim_{h \rightarrow 0}  
   \sup_{|u| \leq h}
   \int_{t_0}^T \partial_e v(s,\bar{E}_s^{t_0,e+u}) ds \biggr) \biggr].
\end{split}
\end{equation*}
Consequently, in order to prove that $(T-t_0) \partial_e v(t_0,e) <1$, it is enough to prove that:
\begin{equation}
 \label{eq:19:04:8}
\lim_{h \rightarrow 0} \sup_{|u| \leq h} \int_{t_0}^T \partial_e v(t,\bar{E}_t^{t_0,e+u}) dt
 \end{equation}
 is finite with non-zero probability. To do so, the Lipschitz bound given by Proposition \ref{prop:1:1} is not sufficient
 since the integral of the bound is divergent.
 To overcome this difficulty, we use \reff{eq:19:4:6}: with non-zero probability, the values of the process
 $(\bar{E}_t)_{t_0 \leq t \leq T}$ at the neighborhood of $T$ may be made
 as large as desired. Precisely, for $A$ as in Proposition \ref{prop:26:08:1}, it is sufficient to prove that there exists $\delta>0$ small enough such that
 ${\mathbb P}[ \inf_{|h| \leq 1} \inf_{T-\delta \leq t \leq T}
 \bar{E}_t^{t_0,e+h} > A ]>0$. For $\delta >0$, we deduce 
 from the boundedness
 of the drift in \reff{eq:05:04:1} that
 \begin{equation*}
 {\mathbb P} \bigl[ \inf_{|h| \leq 1} \inf_{T-\delta \leq t \leq T}
 \bar{E}_t^{t_0,e+h} > A  \bigr]
 \geq {\mathbb P} \biggl[ e- 1 - (T-t_0) + \inf_{T-\delta \leq t \leq T}
\int_{t_0}^t (T-s) dW_s > A  \biggr].
 \end{equation*}
 By independence of the increments of the Wiener integral, we get
 \begin{equation*}
 \begin{split}
&{\mathbb P} \bigl[ \inf_{|h| \leq 1} \inf_{T-\delta \leq t \leq T}
 \bar{E}_t^{t_0,e+h} > A  \bigr]
  \\
 &\geq {\mathbb P} \biggl[ e-1 - (T-t_0) + 
\int_{t_0}^{T-\delta} (T-s) dW_s 
 > 2 A  \biggr] {\mathbb P} \biggl[ \inf_{T-\delta \leq t \leq T} \int_{T-\delta}^t (T-s) dW_s > - A \biggr].
 \end{split}
 \end{equation*}
 The first probability in the above right-hand side is clearly positive for $T-\delta>t_0$. The second one is equal to  
 \begin{equation*}
 {\mathbb P} \biggl[ \inf_{T-\delta \leq t \leq T} \int_{T-\delta}^t (T-s) dW_s > - A \biggr] = 1 -  {\mathbb P} \biggl[ \sup_{T-\delta \leq t \leq T} \int_{T-\delta}^t (T-s) dW_s \geq A \biggr].
 \end{equation*}
 Using maximal inequality (IV.37.12) in Rogers and Williams \cite{RogersWilliams},
the above right hand-side is always positive.
 By \reff{eq:19:4:6}, we deduce that,
 with non-zero probability, the limsup in \reff{eq:19:04:8} is finite. 
\end{proof}

\subsection{Distribution of $\bar{E}_t$ for $t_0 \leq t \leq T$.}

We finally claim:

\begin{proposition}
\label{prop:1:8}
Keep the notation of Propositions \ref{prop:1:6} and \ref{prop:1:7}
and choose some starting point $(t_0,e) \in [0,T) \times \RR$ 
and some $p \in \RR$.
Then, for every $t \in [t_0,T)$, the law of the variable 
\begin{equation*}
E_t^{t_0,e,p} = \bar{E}_t^{t_0,e} - (T-t) P_t^p = 
\bar{E}_t^{t_0,e} - (T-t) \bigl[ p+ W_t \bigr],
\end{equation*}
obtained by transformation
\reff{eq:25:08:1},
is absolutely continuous with respect to the Lebesgue measure. 
At time $t=T$, it has a Dirac mass at $\Lambda$.
\end{proposition}

\begin{proof}
Obviously, we can assume $p=0$, so that $P_t=W_t$. (For simplicity, we will write $E_t^{t_0,e}$ for $E_t^{t_0,e,p}$.) We start with the absolute continuity of $E_t^{t_0,e}$ at time $t<T$. 
Since $v$ is smooth away from $T$, we can compute the Malliavin derivative
of $E_t^{t_0,e}$. (See Theorem 2.2.1 in Nualart \cite{Nualart}.) It satisfies
\begin{equation*}
D_s E_t^{t_0,e} = t-s - \int_s^t \partial_e v \bigl(r,E_r^{t_0,e} +(T-r) W_r \bigr)
D_s E_r^{t_0,e} dr - \int_s^t (T-r) \partial_e v\bigl(r,E_r^{t_0,e} +(T-r) W_r \bigr) dr,
\end{equation*}
for $t_0 \leq s \leq t$. In particular,
\begin{equation}
\label{eq:24:01:10}
\begin{split}
D_s E_t^{t_0,e} &= \int_s^t \biggl[ \bigl[1 - 
(T-r) \partial_e v \bigl(r,E_r^{t_0,e} +(T-r) W_r \bigr)\bigr]
\\
&\hspace{45pt} \times
\exp \biggl( -\int_r^t \partial_e v \bigl(u,E_u^{t_0,e} +(T-u) W_u \bigr) du \biggr) \biggr] dr.
\end{split}
\end{equation}
By Proposition \ref{prop:1:7:b}, we deduce that $D_s E_t^{t_0,e} >0$ for any
$t_0 \leq s \leq t$. By Theorem 2.1.3 in Nualart \cite{Nualart}, we deduce that the law of
$E_t^{t_0,e}$ has a density with respect to the Lebesgue measure.

\vskip 2pt
To prove the existence of a point mass at time $T$, it is enough to
focus on $\bar{E}_T^{t_0,e}$ since the latter is equal to $E_T^{t_0,e}$.  We prove the desired result by comparing the stochastic dynamics of $\bar{E}_T^{t_0,e}$
to the time evolution of solutions of simpler stochastic differential equations. 
With the notation used so far, $\bar{E}_t^{t_0,e}$ is a solution of the stochastic differential equation:
\begin{equation}
\label{fo:sde4E}
d\bar{E}_t=-v(t,\bar{E}_t)dt + (T-t)dW_t
\end{equation}
so it is natural to compare the solution of this equation to solutions of stochastic differential equations with comparable drifts.
Following Remark \ref{rem:24:08:1}, we are going to do so
by comparing $v$ with the solution of the inviscid Burgers equation \reff{eq:21:08:1}.
To this effect we use once more the function $\psi$ defined by $\psi(e) = 1\wedge e^+$ introduced earlier. As said in Remark \ref{rem:24:08:2}, the function $\psi((e-\Lambda)/(T-t))$ is a solution of the Burgers equation 
\reff{eq:21:08:1}
which, up
to the diffusion term (which decreases to $0$ like $(T-t)^2$ when $t\nearrow T$), is the same as the partial differential equation satisfied by $v$.
Using \reff{eq:22:3:1} and \reff{eq:21} with $\Lambda^-=\Lambda^+=\Lambda$ and $\delta=(T-t)^{5/4}$, 
we infer that $v(t,e)$ and $\psi( e-\Lambda/(T-t))$ are exponentially close as $T-t$ tends to $0$
when $e \leq -(T-t)^{5/4}$ or $e \geq T-t + (T-t)^{5/4}$; using \reff{eq:01:04:1} and \reff{eq:01:04:2} with $\Lambda^-=\Lambda^+=\Lambda$ and $\alpha=1/4$, we conclude that the distance between $v(t,e)$ and $\psi((e-\Lambda)/(T-t))$ is at most of order $5/4$ with respect to $T-t$ as $T-t$ tends to $0$
when $-(T-t)^{5/4} < e < T-t + (T-t)^{5/4}$. In any case,  
we have
\begin{equation}
\label{eq:rev:F:1}
\forall e \in \RR, \quad \bigl|v(t,e) - \psi \bigl( \frac{e-\Lambda}{T-t} \bigr) \bigr| \leq C (T-t)^{1/4},
\end{equation}
for some universal constant $C$. We now compare \reff{fo:sde4E} with
\begin{equation}
\label{eq:28:03:1}
d X_t^{\pm} = - \psi \bigl( \frac{X_t^{\pm}-\Lambda}{T-t} \bigr) dt \pm 
C(T-t)^{1/4} dt+ (T-t) dW_t, \quad t_0 \leq t < T,
\end{equation}
with $X_{t_0}^{\pm} = e$ as initial conditions.
Clearly,
\begin{equation}
\label{eq:19:04:11}
X^-_t \leq \bar{E}_t^{e,t_0} \leq X^+_t, \quad t_0 \leq t < T.
\end{equation}
Knowing that $\psi(x)=x$ when $0\le x\le 1$, we anticipate that 
scenarios satisfying $0\le X^\pm_t-\Lambda\le T-t$ can be viewed as solving the stochastic differential equations:
\begin{equation*}
dZ_t^\pm = - \frac{Z_t^\pm-\Lambda}{T-t} dt \pm C (T-t)^{1/4} dt+ (T-t) dW_t,
\end{equation*}
with $Z_{t_0}^\pm=e$ as initial conditions. This remark is useful because these equations have explicit solutions:
\begin{equation}
\label{eq:19:4:10}
Z_t^\pm = \Lambda + (T-t) \bigl[ W_t - W_{t_0} \mp  4C (T-t)^{1/4}
\pm 4C (T-t_0)^{1/4} + \frac{e-\Lambda}{T-t_0} \bigr], \qquad t_0 \leq t \leq T.
\end{equation}
We define the event $F$ by:
$$
F=\left\{ \sup_{t_0\le t\le T}|W_t-W_{t_0}| \le \frac18  \right\}
$$
and we introduce the quantities  $\underline{e}(t_0)$ and $\bar{e}(t_0)$ defined by
$$
\underline{e}(t_0)=\Lambda+\frac14 (T-t_0)\qquad\text{and}\qquad \bar{e}(t_0)=\Lambda+\frac34 (T-t_0)
$$
so that 
$$
\frac14\le \frac{e-\Lambda}{T-t_0}\le \frac34
$$
whenever $\underline{e}(t_0)\le e\le \bar{e}(t_0)$. For such a choice of $e$,  since
$$
\frac{Z^\pm-\Lambda}{T-t} = W_t-W_{t_0}\mp4C(T-t)^{1/4}\pm4C(T-t_0)^{1/4}+\frac{e-\Lambda}{T-t_0},
$$
it is easy to see that if we choose $t_0$ such that $T-t_0$ is small enough for $32C(T-t_0)^{1/4}<1$ to hold, then
\begin{equation*}
\forall t \in [t_0,T], \quad 0\le  \frac{Z_t^--\Lambda}{T-t} \leq \frac{Z_t^+-\Lambda}{T-t} \le 1.
\end{equation*}
on the event $F$. 
This implies that $(X^\pm_t)_{t_0\le t <T}$ and $(Z^\pm_t)_{t_0\le t <T}$ coincide on $F$, and consequently that 
$X_T^+=X_T^-=\Lambda$ and hence $\bar{E}_T^{t_0,e}=\Lambda$ on $F$ by \reff{eq:19:04:11}.
This completes the proof for these particular choices of $t_0$ and $e$. In fact, the result holds for any $e$ and any $t_0 \in [0,T)$. 
Indeed, since $\bar{E}_t^{t_0,e}$ has a strictly positive  
density at any time $t \in (t_0,T)$, if we choose $t_1 \in (t_0,T)$ so that $32C(T-t_1)^{1/4}<1$, then using the Markov property we get
\begin{equation*}
\PP\bigl\{\bar{E}_T^{t_0,e} = \Lambda \bigr\} \geq \int_{\underline{e}(t_1)}^{\bar{e}(t_1)}
\PP\bigl\{\bar{E}_T^{t_1,e'} = \Lambda \bigr\} \PP\bigl\{\bar{E}_{t_1}^{t_0,e} \in de' \bigr\}>0
\end{equation*}
which completes the proof in the general case.
\end{proof}
 
\begin{remark}
We emphasize that the expression for $D_s E_t^{t_0,e}$ given in 
\reff{eq:24:01:10} can vanish with a non-zero probability when replacing $t$ by $T$. Indeed,
the integral 
\begin{equation*}
\int_{r}^T \partial_e v\bigl(u,E_u^{t_0,e} +(T-u) W_u \bigr) du
\end{equation*}
may explode with a non-zero probability since the derivative $\partial_e v(u,e)$ is expected to behave 
like $(T-u)^{-1}$ as $u$ tends to $T$ and $e$ to $\Lambda$. 
Indeed, $v$ is known to behave like the solution of the Burgers equation when close
to the boundary, see \eqref{eq:rev:F:1}. As a consequence, we expect $\partial_e v$ to behave like the gradient of the solution of the Burgers equation. The latter is singular in the neighborhood of the final discontinuity and explodes 
like $(T-u)^{-1}$ in the cone formed by the characteristics of the equation.

However,  in the uniformly elliptic case, the integral above is always bounded since $\partial_e v(u,\cdot)$ is at most of order $(T-u)^{-1/2}$ as explained in Remark \ref{rem:24:08:2}. 
\end{remark}

\subsection{Uniqueness}
Our proof of uniqueness is based on a couple of comparison lemmas.
\begin{lemma}
\label{lem:1}
Let $\phi$ be a non-decreasing smooth function with values in $[0,1]$
greater than ${\mathbf 1}_{[\Lambda,+\infty)}$, and $w$ be the solution of the PDE \reff{eq:3} with $\phi$ as terminal condition. Then,  any solution 
$(\bar{E}_t',Y_t',Z_t')_{t_0 \leq t \leq T}$ of \reff{eq:1:b} starting from $\bar{E}_{t_0}'=e$ 
and satisfying ${\mathbf 1}_{(\Lambda,+\infty)}(\bar{E}_T') \leq Y_T' \leq {\mathbf 1}_{[\Lambda,+\infty)}(\bar{E}_T')$ also satisfies  
\begin{equation*}
w(t,\bar{E}_t') \geq Y_t', \quad t_0 \leq t \leq T.
\end{equation*}
Similarly, if $\phi$ is less than ${\mathbf 1}_{(\Lambda,+\infty)}$, then
\begin{equation*}
w(t,\bar{E}_t') \leq Y_t', \quad t_0 \leq t \leq T.
\end{equation*}
\end{lemma}

\begin{proof}
Applying It\^o's formula to $(w(t,\bar{E}_t')_{t_0 \leq t \leq T}$, we obtain
\begin{equation*}
d \bigl[ w(t,\bar{E}_t') - Y_t' \bigr] = \bigl(  w(t,\bar{E}_t') - Y_t' \bigr)
 \partial_e w(t,\bar{E}_t') dt + \bigl[(T-t)\partial_e w(t,\bar{E}_t')
- Z_t' \bigr] dW_t.
\end{equation*}
Therefore,
\begin{eqnarray*}
&&d \biggl[ \bigl[ w(t,\bar{E}_t') - Y_t' \bigr] \exp \biggl( -\int_{t_0}^t
\partial_e w(s,\bar{E}_s') ds \biggr) \biggr]\\
&&\phantom{????????}
= \exp \biggl( -\int_{t_0}^t \partial_e w(s,\bar{E}_s') ds \biggr)\bigl[(T-t)\partial_e w(t,\bar{E}_t')
- Z_t' \bigr] dW_t.
\end{eqnarray*}
In particular,
\begin{equation*}
w(t,\bar{E}_t') - Y_t' = 
\exp \biggl( \int_{t_0}^t
\partial_e w(s,\bar{E}_s') ds \biggr) {\mathbb E}
\biggl[ \exp \biggl( - \int_{t_0}^T
\partial_e w(s,\bar{E}_s') ds \biggr) \bigl[ w(T,\bar{E}_T') - Y_T' \bigr]
| {\mathcal F}_t \biggr],
\end{equation*}
which completes the proof.
\end{proof}

The next lemma can be viewed as a form of conservation law.

\begin{lemma}
\label{lem:2}
Let $(\chi^n)_{n \geq 1}$ be a non-increasing sequence of non-decreasing smooth functions 
matching $0$ on some intervals $(-\infty,\Lambda^{-,n})_{n \geq 1}$ and $1$ on some intervals $(\Lambda^{+,n},+\infty)_{n \geq 1}$
and
converging towards ${\mathbf 1}_{[\Lambda,+\infty)}$, then the associated
solutions $(w^n)_{n \geq 1}$, given by Proposition \ref{prop:1:1} converge
towards $v$ constructed in Proposition \ref{prop:1:6}. 

The conclusion remains true if  $(\chi^n)_{n \geq 1}$ is a non-decreasing sequence
converging towards ${\mathbf 1}_{(\Lambda,+\infty)}$.
\end{lemma}

\begin{proof} 
Each $w^n$ is a solution of the conservative partial differential equation \reff{eq:3}. Considering
$v^n$ as in the proof of Proposition \ref{prop:1:6}, we have for any $n,m \geq 1$
\begin{equation*}
\int_{\RR} (w^n-v^m)(t,e) de = \int_{\RR} (\chi^n-\phi^m)(e) de, \quad t \in [0,T).
\end{equation*}
Notice that the integrals are well-defined because of Proposition \ref{prop:1:3}.
Since $\phi^m(e) \rightarrow {\mathbf 1}_{[\Lambda,+\infty)}(e)$ as $m
\rightarrow + \infty$ for $e \not = \Lambda$, we deduce that
\begin{equation*}
\int_{\RR} (w^n-v)(t,e) de = 
\int_{\RR} \bigl[\chi^n(e) - {\mathbf 1}_{[\Lambda,+\infty)}(e)\bigr] de, \quad t \in [0,T).
\end{equation*}
Since the right hand side converges towards 0 as $n$ tends to $+\infty$, so does the left hand side,
but since $w^n(t,e) \geq v(t,e)$ by Lemma \ref{lem:1} (choosing $(\bar{E}',Y',Z')=(\bar{E}^{t_0,e},Y^{t_0,e},Z^{t_0,e})$), we must also have:
\begin{equation*}
\lim_{n \rightarrow + \infty} \int_{\RR} |w^n(t,e)-v(t,e)| de = 0.
\end{equation*}
Since $(w^n(t,\cdot))_{n \geq 1}$ is equicontinuous (by Proposition 
\ref{prop:1:1}), we conclude that $w^n(t,e) \rightarrow v(t,e)$. 
The proof is similar if $\chi^n \nearrow {\mathbf 1}_{(\Lambda,+\infty)}$.
\end{proof}

\vskip 6pt
To complete the proof of uniqueness, consider a sequence $(\chi^n)_{n \geq 1}$
as in the statement of Lemma \ref{lem:2}. For any solution
$(\bar{E}_t',Y_t',Z_t')_{t_0 \leq t \leq T}$ of \reff{eq:1:b} with $\bar{E}_{t_0}'=e$, Lemma \ref{lem:1} yields
\begin{equation*}
w^n(t,\bar{E}_t') \geq Y_t', \quad t \in [t_0,T).
\end{equation*}
Passing to the limit, we conclude that
\begin{equation*}
v(t,\bar{E}_t') \geq Y_t', \quad t \in [t_0,T).
\end{equation*}
Choosing a non-decreasing sequence $(\chi^n)_{n \geq 1}$, instead, we obtain the reverse inequality, and hence, we conclude that $Y_t' = v(t,\bar{E}_t')$ 
for $t \in [t_0,T)$. By uniqueness
to \reff{eq:05:04:1}, we deduce that $\bar{E}_t'= \bar{E}_t^{t_0,e}$, so that
$Y_t'=Y_t^{t_0,e}$. We easily deduce that $Z_t'=Z_t^{t_0,e}$ as well. 
\vskip 2pt

\begin{remark}
We conjecture that the analysis performed in this section can be extended to more general conservation laws than Burgers 
equation.  The Burgers case is
the simplest one since the corresponding forward - backward stochastic differential equation is purely linear.
\end{remark}

\section{\textbf{Option Pricing and Small Abatement Asymptotics}}
\label{se:option}
In this section, we consider the problem of option pricing in the framework of the first equilibrium model introduced in this paper. 

\subsection{PDE Characterization}

Back to the risk neutral dynamics of the (perceived) emissions given by \reff{fo:RNfbsde}, we assume that the emissions of the business as usual scenario are modeled by a geometric Brownian motion, so that $b(t,e)=b e$ and $\sigma(t,e)=\sigma e$. As explained in the introduction, this model has been used in most of the early reduced form analyses of emissions allowance forward contracts and option prices (see \cite{ChesneyTaschini} and \cite{CarmonaHinz} for example). The main thrust of this section is to include the impact of the allowance price $Y$ on the dynamics of the cumulative emissions. As we already saw in the previous section, this feedback  $f(Y_s)$ is the source of a nonlinearity in the PDE whose solution determines the price of an allowance. Throughout this section, we assume that under the pricing measure (martingale spot measure) the cumulative emissions and the price of a forward contract on an emission allowance satisfy the forward-backward system:
 \begin{equation}
 \label{fo:BMfbsde}
 \begin{cases}
&\displaystyle E_t= E_0+\int_0^t(bE_s-f(Y_s))ds+\int_0^t \sigma E_s d W_s \\
&\displaystyle Y_t = \lambda\;\mathbf{1}_{[\Lambda,\infty)}(E_T) -\int_t^T Z_td W_t,
 \end{cases}
 \end{equation}
with $f$ as in \eqref{f-Niz} with $f(0)=0$ and $\lambda,\Lambda >0$. For notational convenience, the martingale measure is denoted by $\PP$ instead of ${\mathbb Q}$ as in Section \ref{se:equilibrium} and the associated Brownian motion by $(W_{t})_{0 \leq t \leq T}$
instead of $(\tW_{t})_{0 \leq t \leq T}$.

Theorem \ref{thmsec3} directly applies here, so that equation \eqref{fo:BMfbsde} is uniquely solvable given the initial condition 
$E_{0}$. In particular, we know from the proof of Theorem \ref{thmsec3} that the solution $(Y_t)_{0 \leq t \leq T}$ of the backward equation is constructed as a function $(Y_t=u(t,E_t))_{0 \leq t \leq T}$ of the solution of the forward equation. 
Moreover, since we are assuming that $f(0)=0$, it follows from Proposition \ref{prop:sign} that the process $E$ takes positive values.
 
Referring to \cite{PardouxPeng}, we notice that the function $u$ is the right candidate for being the viscosity solution to the PDE 
 \begin{equation}
 \label{fo:bsde->pde}
 \begin{cases}
 &\partial_t u(t,e) + (be-f(u(t,e)))\partial_eu(t,e) + \frac12\sigma^2e^2 \partial_{ee}^2u(t,e) =0,\qquad (t,e)\in[0,T)\times\RR_+\\
& u(T,.)=\lambda\mathbf{1}_{[\Lambda,\infty)}
\end{cases}
\end{equation}
Having this connection in mind, we consider next the price at time $t<\tau$ of a European call option with maturity $\tau<T$ and strike $K$ on an allowance forward contract maturing at time $T$. It is given by the expectation
 $$
\EE\{(Y_\tau^{t,e}-K)^+\}=\EE\{(u(\tau,E_\tau^{t,e})-K)^+\},
 $$ 
which can as before, be written as a function $U(t,E_t^{t,e})$ of the current value of the cumulative emissions, where the notation 
$(t,e)$ in superscript indicates that $E_t=e$. Once the function $u$ is known and/or computed, for exactly the same reasons as above, the function $U$ appears as the viscosity solution of the linear partial differential equation:
 \begin{equation}
 \label{fo:BSpde}
\begin{cases}
& \partial_t U(t,e) + (be-f(u(t,e)))\partial_eU(t,e) + \frac12\sigma^2e^2 \partial_{ee}^2U(t,e) =0,\quad (t,e)\in[0,\tau)\times\RR_+\\
&U(\tau,.)=(u(\tau,.)-K)^+, 
\end{cases}
\end{equation}
which, given the knowledge of $u$, is a linear partial differential equation. Notice that in the case $f\equiv 0$ of infinite abatement costs, except for the fact that the coefficients of the geometric Brownian motion were assumed to be time dependent, the above option price is the same as the one derived in \cite{CarmonaHinz}.

%is a {\color{red} viscosity} solution of the nonlinear partial differential equation} 
% \begin{equation}
% \label{fo:bsde->pde}
% \begin{cases}
% &\partial_t v(t,e) + (be-f(v(t,e)))\partial_ev(t,e) + \frac12\sigma^2e^2 \partial_{ee}^2v(t,e) =0,\qquad (t,e)\in[0,T)\times\RR\\
%& v(T,.)=\lambda\mathbf{1}_{[\Lambda,\infty)}.
%\end{cases}
%\end{equation}
%{\color{red} Notice that, for the sake of simplicity, we are not imposing that the  emissions process $E$ is non-negative, see %Proposition \ref{prop:sign}.}

\subsection{Small Abatement Asymptotics}

Examining the PDEs \reff{fo:bsde->pde} and \reff{fo:BSpde}, we see that there are two main differences with the classical Black-Scholes framework. First, the underlying contract price is determined by the nonlinear PDE \reff{fo:bsde->pde}. Second, the option pricing PDE \reff{fo:BSpde} involves the nonlinear term $f(u(t,e))$, while still being linear in terms of the unknown function $U$. Because the function $u$ is determined by the first PDE \reff{fo:bsde->pde}, this nonlinearity is inherent to the model, and one cannot simply reduce the PDE to the Black-Scholes equation.

\vskip 4pt
In order to understand the departure of the option prices from those of the Black-Scholes model, we introduce a small parmater $\epsilon\ge 0$, and take the abatement rate to be of the form $f=\epsilon f_0$ for some fixed non-zero increasing continuous function $f_0$. We denote by $u^\epsilon$ and $U^\epsilon$ the corresponding prices of the allowance forward contract and the option.  Here, what we call \emph{Black-Scholes model} corresponds to the case $f\equiv 0$. Indeed, in this case, both \reff{fo:bsde->pde} and \reff{fo:BSpde} reduce to the linear Black-Scholes PDE, differing only through their boundary conditions. This model was one of the models used in \cite{CarmonaHinz} for the purpose of pricing options on emission allowances based on price data exhibiting no implied volatility smile.

For $\epsilon=0$, the nonlinear feedback given by the abatement rate disappears and we easily compute that, for $e>0$,
 \begin{eqnarray}
 u^0(t,e) &=& \lambda\EE\left[{\mathbf 1}_{[\Lambda,\infty)}(E^{0,t,e}_T)\right]
 \;=\; \lambda\Phi\left(\frac{\ln[e \exp(b(T-t))/\Lambda ]}{\sigma\sqrt{T-t}} - \frac{\sigma\sqrt{T-t}}{2}\right)
 \label{v0}\\
 U^0(t,e) &=& \EE\left[(u^0(\tau,E^{0,t,e}_\tau)-K)^+\right],\qquad 0\le t\le \tau,
 \label{V0}
 \end{eqnarray}
where $E^{0,t,e}$ is the geometric Brownian motion:
 \begin{equation}
 \label{fo:gbm}
 dE_s^{0,t,e}=E^{0,t,e}_s [ b ds + \sigma d W_s], \quad s \geq t,
 \end{equation}
used as a proxy for the cumulative emissions in business as usual, with the initial condition $E_{t}^{t,e}=e$.
See for example \cite{CarmonaHinz} for details and complements. The main technical result of this section is the following first order Taylor expansion of the option price.
 
\begin{proposition}\label{propTaylor1}
Let $f$ satisfy \eqref{f-Niz} and $(t,e) \in [0,\tau) \times (0,+\infty)$. Then, as $\epsilon\to 0$, we have
\begin{equation*}
\begin{split}
 &U^\epsilon(t,e) = U^0(t,e)
 \\
 &\hspace{5pt}  - \epsilon\;\EE\left[\mathbf{1}_{[\Lambda,\infty)}(u^0(\tau,E^{0,t,e}_\tau))
                       \int_t^T f_0(u^0(s,E^{0,t,e}_s))
                                \partial_e u^0(s\vee\tau,E^{0,t,e}_{s\vee\tau})
                                \frac{E^{0,t,e}_{s\vee\tau}}
                                     {E^{0,t,e}_s}
                                ds
                \right]
 +o(\epsilon),
 \end{split}
 \end{equation*}
where $\epsilon^{-1}o(\epsilon)\longrightarrow 0$ as $\epsilon\to 0$.
\end{proposition}

\begin{proof} 
The proof divided into four parts.
\vskip2pt

(i) We first prove that the functions $u^0$ and $U^0$, with $u^0\equiv 0$ and $U^0 \equiv 0$ on 
$[0,T] \times \RR_{-}$ and $[0,\tau] \times \RR_{-}$ respectively, belongs to ${\mathcal C}^{1,2}([0,T) \times \RR)$ and 
${\mathcal C}^{1,2}([0,\tau) \times \RR)$ respectively. 

By \eqref{v0}, we know that $u^0$ is ${\mathcal C}^{1,2}$ on $[0,T) \times \RR_{+}^*$. Obviously $u^0 \equiv 0$
on the whole $[0,T] \times \{0\}$ since $\Lambda>0$. Using the bound
\begin{equation*}
\int_{-\infty}^{-x} \exp \bigl( - \frac{v^2}{2}\bigr) \frac{dv}{\sqrt{2 \pi}} \leq 
\sqrt{\frac{2}{\pi}} x^{-1} \exp\bigl( - \frac{x^2}{2} \bigr), \quad x >0,
\end{equation*}
we deduce that 
\begin{equation*}
u^0(t,e) \leq \sqrt{\frac{2}{\pi}} \lambda \frac{\sigma (T-t)^{1/2}}{\vert \ln(e \exp(b(T-t))/\Lambda) \vert}
\exp \bigl( - \frac{\vert \ln(e \exp(b(T-t))/\Lambda) \vert^2}{2 \sigma^2 (T-t)} \bigr),
\end{equation*}
for $0 < e  \ll 1, \ t \in [0,T)$. This shows that $u^0(t,e)$ decays towards $0$ faster than any polynomial. In particular
$\partial_{e} u^0(t,0)=\partial_{ee}^2 u^0(t,0)=0$. Differentiating \eqref{v0} with respect to $e$, we conclude by the same argument that 
$\partial_{e} u^0(t,e)$ and $\partial_{ee}^2 u^0(t,e)$ decay towards $0$ faster than any polynomial, so that 
the first and second-order derivatives in space are continuous on $[0,T) \times \RR_{+}$. Obviously, $\partial_{t} u^0(t,0)=0$ for any $t \in [0,T)$ and, by differentiating \eqref{v0} with respect to $t$, we can also prove that $\partial_{t} u^0$ is continuous on $[0,T) \times \RR_{+}$. Since $u^0 \equiv 0$ on $[0,T] \times \RR_{-}^*$, we deduce that $u^0$ is of class 
${\mathcal C}^{1,2}$ on $[0,T) \times \RR$.

All in all, the computation of the first-order derivatives yields
\begin{equation*}
\partial_{e} u^0(t,e) = \frac{\lambda}{\sqrt{2\pi} e \sigma (T-t)^{1/2}}
\exp \bigl( - \frac{\vert \ln(e \exp(b(T-t))/\Lambda)  - \sigma^2 (T-t)/2\vert^2}{2 \sigma^2 (T-t)} \bigr), 
\end{equation*}
for $e >0$ and $t \in [0,T)$. The above right-hand side is less than $C/(T-t)^{1/2}$ for $e$ away from $0$, the constant $C$ being independent of $t$. 
When $e$ is close to $0$, 
\begin{equation*}
\vert \partial_{e} u^0(t,e) \vert \leq 
\frac{\lambda}{\sqrt{2\pi} e \sigma (T-t)^{1/2}}
\exp \bigl( - \frac{\vert \ln(e \exp(b(T-t))/\Lambda) \vert^2}{2 \sigma^2 (T-t)} \bigr),
\end{equation*}
so that the bound 
\begin{equation}
\label{eq:6:12:1}
\vert \partial_{e} u^0(t,e) \vert \leq C (T-t)^{-1/2}, \quad t \in [0,T), \ e \in \RR,
\end{equation}
is always true. As a by-product, we deduce that 
$u^0(\tau,E_{\tau}^0) \leq C(T-\tau)^{-1/2} \vert E_{\tau}^0 \vert$,
so that 
\begin{equation*}
(u^0(\tau,E_{\tau}^0)-K)^+ = (u^0(\tau,E_{\tau}^0)-K)^+ {\mathbf 1}_{\{
\vert E_{\tau}^0 \vert \geq (T-\tau)^{1/2}K/C \}}.
\end{equation*}
In particular,
\begin{equation*}
U^0(t,e) = {\mathbb E} \bigl[(u^0(\tau,E_{\tau}^{0,t,e})-K)^+ {\mathbf 1}_{\{
\vert E_{\tau}^{0,t,e} \vert \geq (T-\tau)^{1/2}K/C \}} \bigr]. 
\end{equation*}
By the same argument as the one used for $u^0$, we see that $U^0$ and its partial derivatives with respect to $t$ and $e$ decay towards $0$ as $e$ tends to $0$, at a faster rate than any polynomial one. In particular, setting $U^0(t,e)=0$ for $t \in [0,\tau]$ and 
$e \in \RR$, we deduce that $U^0$ belongs to ${\mathcal C}^{1,2}([0,\tau) \times \RR)$.

(ii) We use the smoothness of $u^0$ and apply It\^o's formula to $(u^0(s,E_{s}^{\epsilon,t,e}))_{t \leq s \leq T}$, 
where $E^{\epsilon,t,e}$ denotes the forward process in \eqref{fo:BMfbsde}, when $f \equiv \epsilon f_{0}$ and 
under the initial condition $E^{\epsilon,t,e}_{t}=e>0$, $t \in [0,\tau)$. Using the fact that $u^0$ belongs to ${\mathcal C}^{1,2}([0,T) \times \RR)$ together with (\ref{eq:6:12:1}), we deduce that, for any $t \leq S <T$,
\begin{equation*}
u^0(t,e) = {\mathbb E} \biggl[ 
u^0(S,E_{S}^{\epsilon,t,e})  + \epsilon \int_{t}^S f_{0}\bigl( u^{\epsilon}(s,E_{s}^{\epsilon,t,e})\bigr) \partial_{e} u^0(s,E_{s}^{\varepsilon,t,e}) ds \biggr]. 
\end{equation*}
Clearly, $u^0(S,e) \longrightarrow \lambda {\mathbf 1}_{[\Lambda,+\infty)}(e)$ as $S \nearrow T$, for $e \not = S$. Since
${\mathbb P}[E_{T}^{\epsilon,t,e} = \Lambda ]=0$, see (ii-3) in the proof of Theorem \ref{thmsec3}, we deduce that 
${\mathbb E} u^0(S,E_{S}^{\epsilon,t,e}) \longrightarrow \lambda {\mathbb E}{\mathbf 1}_{[\Lambda,+\infty)}(E_{T}^{\epsilon,t,e})=u^\epsilon(t,e)$
as $S \nearrow T$. Therefore,
\begin{equation}
\label{eq:6:12:2}
u^0(t,e) - u^{\epsilon}(t,e) =  \epsilon {\mathbb E}   \int_{t}^T f_{0}\bigl( u^{\epsilon}(s,E_{s}^{\epsilon,t,e})\bigr) \partial_{e} u^0(s,E_{s}^{\epsilon,t,e}) ds, 
\end{equation}
the right-hand side above making sense because of the bound \eqref{eq:6:12:1}. By a similar argument, we get
\begin{equation}
\label{eq:6:12:3}
U^0(t,e)  = {\mathbb E} \bigl( u^0(\tau,E_{\tau}^{\epsilon,t,e}) - K \bigr)^+ + \epsilon {\mathbb E}  \int_{t}^{\tau} f_{0}\bigl( u^{\epsilon}(s,E_{s}^{\epsilon,t,e})\bigr) \partial_{e} U^0(s,E_{s}^{\epsilon,t,e}) ds. 
\end{equation}
Notice that $\partial_{e} U^0$ is bounded since $u^0(\tau,\cdot)$ is Lipschitz-continuous so that the integral above is well-defined.
By \eqref{eq:6:12:1} and \eqref{eq:6:12:2}, we know that $\|u^0(\tau,\cdot) - u^{\epsilon}(\tau,\cdot)\|_{\infty} \leq C \epsilon$, for a constant $C$ independent of $\epsilon$ and $\tau$. Therefore,
\begin{equation*}
\begin{split}
&{\mathbb E} \bigl( u^0(\tau,E_{\tau}^{\epsilon,t,e}) - K \bigr)^+ 
\\
&= {\mathbb E} \bigl[ \bigl( u^0(\tau,E_{\tau}^{\epsilon,t,e}) - K \bigr)^+ {\mathbf 1}_{\{
\vert 
u^0(\tau,E_{\tau}^{\epsilon,t,e}) - K \vert \geq C \epsilon \}} \bigr]
\\
&\hspace{15pt} + 
{\mathbb E} \bigl[ \bigl( u^0(\tau,E_{\tau}^{\epsilon,t,e}) - K \bigr)^+ {\mathbf 1}_{\{
\vert 
u^0(\tau,E_{\tau}^{\epsilon,t,e}) - K \vert < C \epsilon\}} \bigr]
\\
&= {\mathbb E} \bigl[ \bigl( u^{\epsilon}(\tau,E_{\tau}^{\epsilon,t,e}) - K \bigr)^+ {\mathbf 1}_{\{
\vert 
u^0(\tau,E_{\tau}^{\epsilon,t,e}) - K \vert \geq C \epsilon \}} \bigr]
\\
&\hspace{15pt} + \epsilon {\mathbb E} \biggl[   {\mathbf 1}_{\{
u^0(\tau,E_{\tau}^{\epsilon,t,e}) \geq  K + C \epsilon  \}}
\int_{\tau}^T f_{0}\bigl( u^{\epsilon}(s,E_{s}^{\epsilon,t,e})\bigr) \partial_{e} u^0(s,E_{s}^{\epsilon,t,e}) ds \biggr]
\\
&\hspace{15pt} + 
{\mathbb E} \bigl[ \bigl( u^0(\tau,E_{\tau}^{\epsilon,t,e}) - K \bigr)^+ {\mathbf 1}_{\{
\vert 
u^0(\tau,E_{\tau}^{\epsilon,t,e}) - K \vert < C \epsilon \}} \bigr]
\\
&= {\mathbb E} \bigl[ \bigl( u^{\epsilon}(\tau,E_{\tau}^{\epsilon,t,e}) - K \bigr)^+  \bigr]
 + \epsilon {\mathbb E} \biggl[   {\mathbf 1}_{\{
u^0(\tau,E_{\tau}^{\epsilon,t,e})  \geq K \}}
\int_{\tau}^T f_{0}\bigl( u^{\epsilon}(s,E_{s}^{\epsilon,t,e})\bigr) \partial_{e} u^0(s,E_{s}^{\epsilon,t,e}) ds \biggr]
\\
&\hspace{15pt} + \epsilon O \bigl( {\mathbb P}[\vert 
u^0(\tau,E_{\tau}^{\epsilon,t,e}) - K \vert \leq C \epsilon ] \bigr),
\end{split}
\end{equation*}
where $O(\cdot)$ stands for the Landau notation. By \eqref{eq:6:12:3}, we finally get
\begin{equation}
\label{eq:6:12:4}
\begin{split}
U^0(t,e) &= U^{\epsilon}(t,e) + \epsilon {\mathbb E} \biggl[   {\mathbf 1}_{\{
u^0(\tau,E_{\tau}^{\epsilon,t,e})  \geq K \}}
\int_{\tau}^T f_{0}\bigl( u^{\epsilon}(s,E_{s}^{\epsilon,t,e})\bigr) \partial_{e} u^0(s,E_{s}^{\epsilon,t,e}) ds \biggr]
\\
&\hspace{15pt} + \epsilon {\mathbb E} \int_{t}^{\tau} f_{0}\bigl( u^{\epsilon}(s,E_{s}^{\epsilon,t,e})\bigr) \partial_{e} U^0(s,E_{s}^{\epsilon,t,e}) ds
\\
&\hspace{15pt} + \epsilon O \bigl( {\mathbb P}[
\vert 
u^0(\tau,E_{\tau}^{\epsilon,t,e}) - K \vert \leq C \epsilon]\bigr).
\end{split}
\end{equation}

(iii) We now prove that: 
\begin{equation*}
\lim_{\epsilon \rightarrow 0} {\mathbb P}\vert 
u^0(\tau,E_{\tau}^{\epsilon,t,e}) - K \vert \leq C \epsilon] = 0
\end{equation*}
for $(t,e) \in [0,\tau) \times (0,+\infty)$. For any $\delta >0$, 
\begin{equation*}
\limsup_{\epsilon \rightarrow 0}
{\mathbb P}[\vert u^0(\tau,E_{\tau}^{\epsilon,t,e}) - K \vert \leq C \epsilon ]
\leq 
\limsup_{\epsilon \rightarrow 0}
 {\mathbb P}[\vert u^0(\tau,E_{\tau}^{\epsilon,t,e}) - K \vert \leq C \delta ].
\end{equation*}
By continuity with respect to parameters of solutions of stochastic differential equations, we see that that $E_{\tau}^{\epsilon,t,e} \longrightarrow E_{\tau}^{0,t,e}$
a.s. as $\epsilon \rightarrow 0$. Therefore, by the porte-manteau theorem, it holds for any $\delta >0$,
\begin{equation*}
\limsup_{\epsilon \rightarrow 0}
{\mathbb P}[\vert u^0(\tau,E_{\tau}^{\epsilon,t,e}) - K \vert \leq C \epsilon ]
\leq  {\mathbb P}[\vert u^0(\tau,E_{\tau}^{0,t,e}) - K \vert \leq C \delta ].
\end{equation*}
On the interval $[K-C \delta,K+C \delta ]$, with $\delta$ small enough so that 
$K-C \delta >0$,
 the function 
$u^0(\tau,\cdot)$ is continuously differentiable with a non-zero derivative and thus defines a $C^1$-diffeomorphism. Moreover, since $e > 0$, the random variable 
$E_{\tau}^{0,t,e}$ has a smooth density on the interval 
$[K-C \delta,K+C \delta]$.
 Therefore, the random
variable $u^0(\tau,E_{\tau}^{0,t,e})$ has a continuous density on the interval 
$[K-C \delta,K+C \delta]$. We conclude that 
\begin{equation*}
\lim_{\delta \rightarrow 0}
{\mathbb P}[\vert u^0(\tau,E_{\tau}^{0,t,e}) - K \vert \leq C \delta] = 0.
\end{equation*}
(iv) We now have all the ingredients needed to complete the proof.
From \eqref{eq:6:12:4}, we have:
\begin{equation*}
\begin{split}
U^0(t,e) &= U^{\epsilon}(t,e) + \epsilon {\mathbb E} \biggl[   {\mathbf 1}_{\{
u^0(\tau,E_{\tau}^{\epsilon,t,e})  \geq K \}}
\int_{\tau}^T f_{0}\bigl( u^{\epsilon}(s,E_{s}^{\epsilon,t,e})\bigr) \partial_{e} u^0(s,E_{s}^{\epsilon,t,e}) ds \biggr]
\\
&\hspace{15pt} + \epsilon {\mathbb E} \biggl[   \int_{t}^{\tau} f_{0}\bigl( u^{\epsilon}(s,E_{s}^{\epsilon,t,e})\bigr) \partial_{e} U^0(s,E_{s}^{\epsilon,t,e}) ds \biggr]  + \epsilon o(\epsilon).
\end{split}
\end{equation*}
Since, for any $s \in [0,T)$, $u^{\epsilon}(s,\cdot)$ converges towards $u^0(s,\cdot)$ uniformly as
$\epsilon$ tends to 0, and since $\PP[u^0(\tau,E_{\tau}^{0,t,e})=K]=0$, we deduce from \eqref{eq:6:12:1} and from Lebesgue dominated convergence theorem that
\begin{equation*}
\begin{split}
U^0(t,e) &= U^{\epsilon}(t,e) + \epsilon {\mathbb E} \biggl[   {\mathbf 1}_{\{
u^0(\tau,E_{\tau}^{0,t,e})  \geq K \}}
\int_{\tau}^T f_{0}\bigl( u^{0}(s,E_{s}^{0,t,e})\bigr) \partial_{e} u^0(s,E_{s}^{0,t,e}) ds \biggr]
\\
&\hspace{15pt} + \epsilon {\mathbb E} \biggl[   \int_{t}^{\tau} f_{0}\bigl( u^{0}(s,E_{s}^{0,t,e})\bigr) \partial_{e} U^0(s,E_{s}^{0,t,e}) ds \biggr]  + \epsilon o(\epsilon).
\end{split}
\end{equation*}
The final result then follows from the identity:
$$
\partial_e U^0(t,e)
=
\EE\left[\frac{E^{0,t,e}_\tau}{E^{0,t,e}_t}
               \mathbf{1}_{[K,\infty)}(u^0(\tau,E^{0,t,e}_\tau))
                         \partial_e u^0(\tau,E^{0,t,e}_\tau)
                   \right], \quad 0 \leq t < \tau, \ e >0,
$$
which can be derived by differentiation of  \eqref{V0}
 and making use of the  equality 
${\mathbb P}[u^0(\tau,E_{\tau}^{0,t,e}) =K]=0$.\end{proof}

\subsection{Numerical results}
\label{sub:numerics}

In this final subsection we provide the following numerical evidence of the accuracy of the small abatement asymptotic formula derived above:
\begin{enumerate}
\item We compute numerically $u^\epsilon$ with high accuracy, and we then compute values of $U^\epsilon$ using the values of $u^\epsilon$ so computed.
We used an explicit finite difference monotone scheme (see for example \cite{BarlesSouganidis} for details). The left pane of Figure \ref{fi:option_prices} gives a typical sample of results. For the sake of illustration we used the abatement function $f(x)=x$ corresponding to quadratic costs of abatement. The penalty, cap, emission volatility and emission rate in BAU  were chosen as
$\lambda=1$, $\Lambda=1.25$, $\sigma=0.3$ and $b=2\Lambda/T$ where the length of the regulation period was $T=1$ year.
The prices of the allowances $u(t,e)$ and $u^\epsilon(t,e)$ were computed on a regular grid in the time $\times$ log-emission space.
The mesh of the time subdivision was $\Delta t=1/250$. The grid of $1001$ log-emission was regular, centered around $0$ with mesh $\Delta x$ connected to $\Delta t$ by the standard stability condition. We considered an option with maturity $\tau=0.25$ and strike $K=0.86$. We computed $u^\epsilon(t,e)$ and $U^\epsilon(t,e)$ over this grid for $11$ values of $\epsilon$, $\epsilon=0, 0.1, 0.2, \cdots, 0.9, 1$, and we plotted the option prices $U^\epsilon(t,e)$ against the corresponding allowance prices $u^\epsilon(t,e)$. The graphs decrease as $\epsilon$ increases from $0$ to $1$.

\begin{figure}
\centerline{ 
\includegraphics[height=6cm, width=8cm]{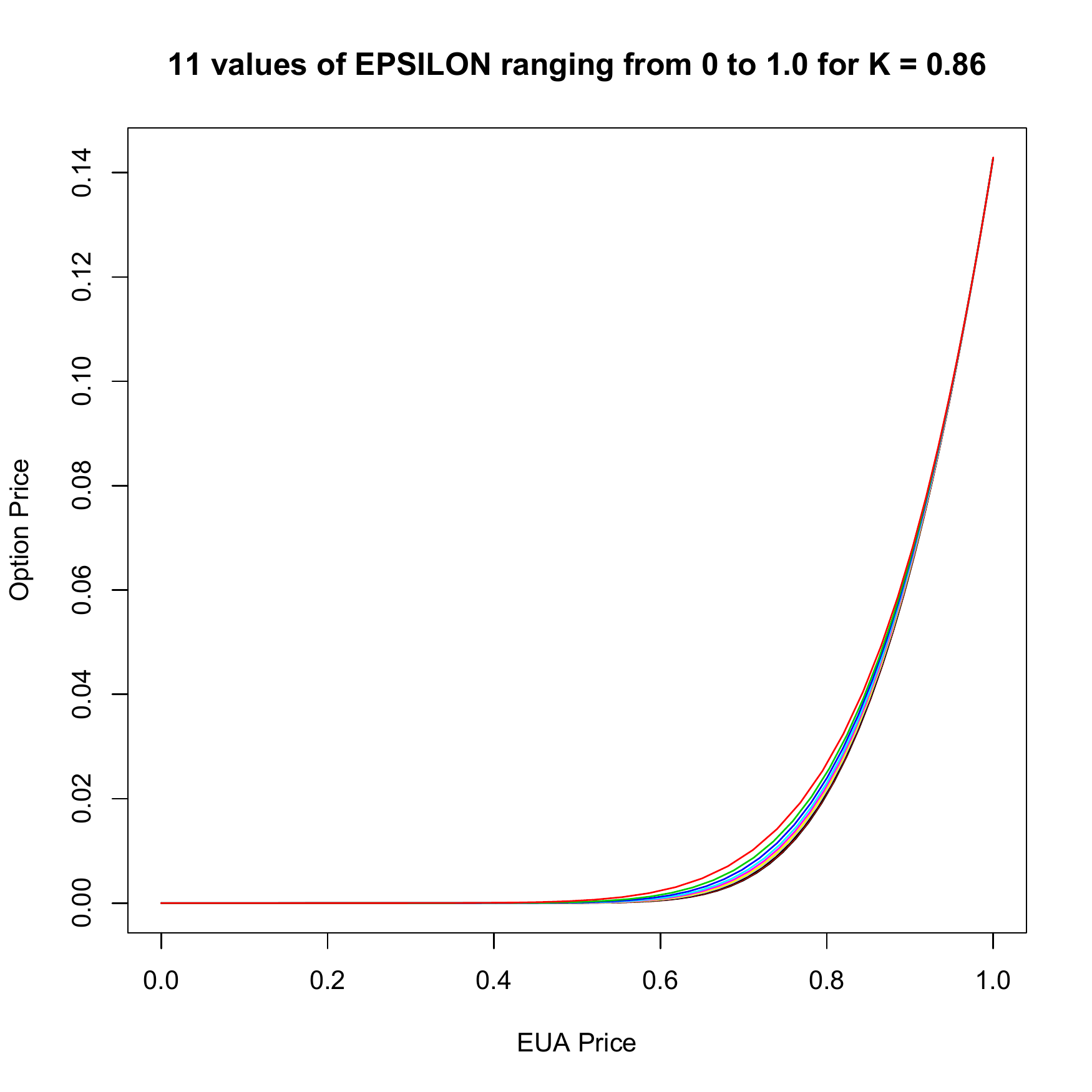}
\hskip 6pt
\includegraphics[height=6cm, width=8cm]{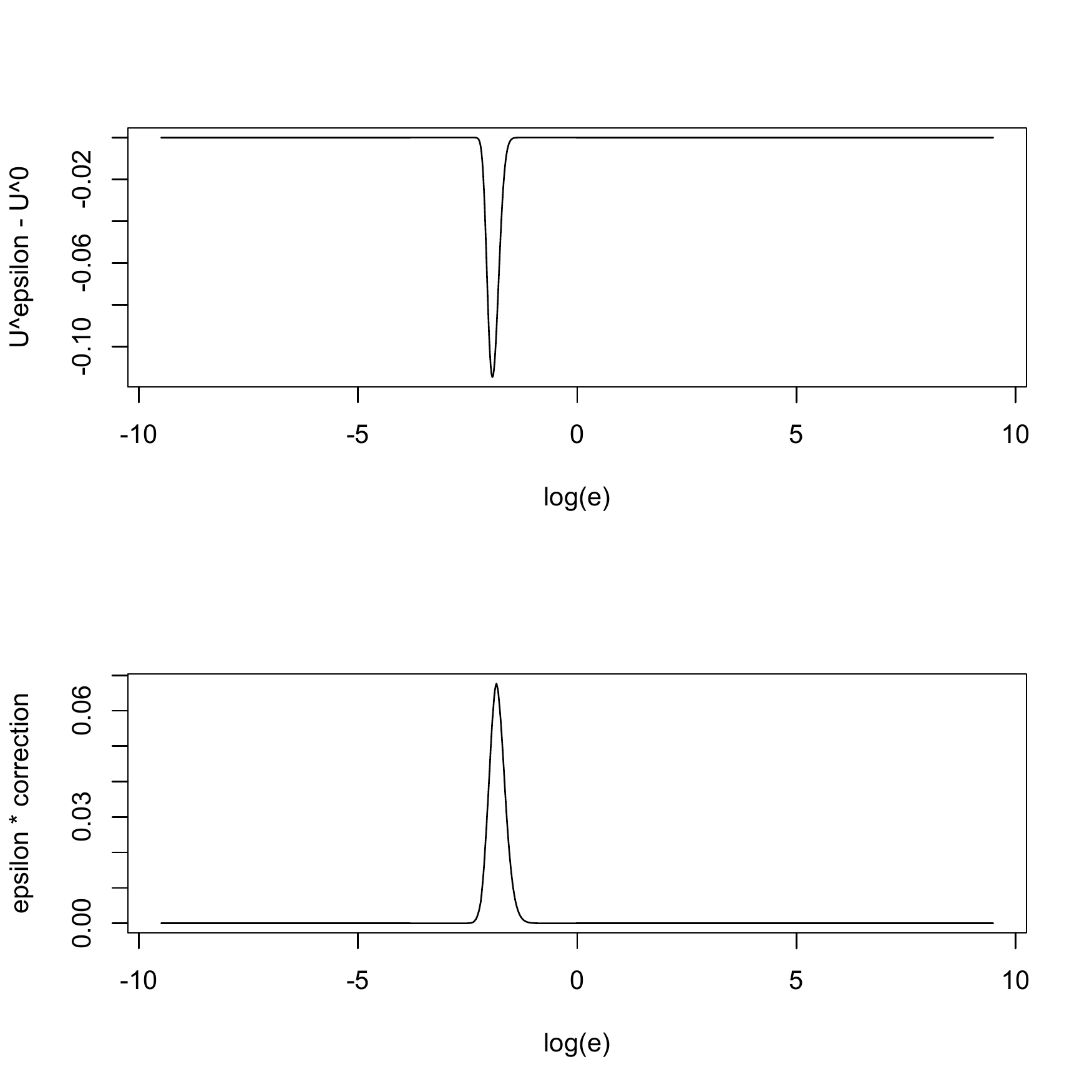}
}
\caption{European call option prices for $\epsilon=0,0.1,0.2,\cdots,0.9,1$.}
\label{fi:option_prices}
\end{figure}
We plotted $U^\epsilon$ against $u^\epsilon$ in order to show how the option price depends upon the value of the underlying allowance.

\item We also computed the expectation appearing as the coefficient of $\epsilon$ in the first order expansion of Proposition \ref{propTaylor1}. We used a plain  Monte Carlo computation of the expectation with $N=10,000$ sample paths. The right pane of 
Figure \ref{fi:option_prices} shows the potential of the approximation for $\epsilon=0.1$. The top plot shows the difference between 
the exact option value and the linear approximation given by setting $\epsilon=0$ and ignoring the feedback effect. Both option values were computed by solving the partial differential equations as explained at the beginning of the section. The lower plot shows the first order correction as identified in Proposition \ref{propTaylor1}, showing the potential of the approximation.
\end{enumerate}

\bibliographystyle{plain}

\end{document}